%% file: main.tex
\pdfoutput=1
\documentclass[runningheads]{llncs}
\usepackage[T1]{fontenc}
\usepackage{graphicx}
\usepackage{orcidlink}
\usepackage{hyperref}
\hypersetup{
    colorlinks=true,
    linkcolor=blue,
    filecolor=blue,      
    urlcolor=blue,
    citecolor=blue,
}

\usepackage{amsmath}
\usepackage{amssymb,amsfonts}
\usepackage{verbatim}  %for \begin{comment}\end{comment}
\usepackage{stmaryrd}
\usepackage[inline]{enumitem}
\usepackage{mathpartir}
\usepackage{xspace}
\usepackage{subfigure}
\usepackage{macros}
\usepackage{algorithm}
\usepackage{algorithmic}
\usepackage{multicol}
\usepackage{adjustbox}
\usepackage{booktabs}
\usepackage{listings}
\usepackage{wrapfig}
\usepackage{siunitx}

\lstdefinestyle{compactC}{
  language=C,
  basicstyle=\ttfamily\small,
  numbers=left,
  numberstyle=\tiny,
  stepnumber=1,
  numbersep=5pt,
  tabsize=1,
  xleftmargin=2pt,
  xrightmargin=2pt,
  showspaces=false,
  showstringspaces=false,
  breaklines=true
}

\makeatletter
\providecommand*{\cupdot}{%
  \mathbin{%
    \mathpalette\@cupdot{}%
  }%
}
\newcommand*{\@cupdot}[2]{%
  \ooalign{%
    $\m@th#1\cup$\cr
    \hidewidth$\m@th#1\cdot$\hidewidth
  }%
}
\providecommand*{\bigcupdot}{%
  \mathbin{%
    \mathpalette\@bigcupdot{}%
  }%
}
\newcommand*{\@bigcupdot}[2]{%
  \ooalign{%
    $\m@th#1\bigcup$\cr
    \hidewidth$\m@th#1\cdot$\hidewidth
  }%
}
\makeatother
\begin{document}
\title{Bit-Precise CHC Satisfiability Using Theory-Modular Reasoning}

%
%\titlerunning{Abbreviated paper title}
% If the paper title is too long for the running head, you can set
% an abbreviated paper title here
%
\author{Omer Rappoport\,\orcidlink{0009-0003-1518-4033} \and
Orna Grumberg\,\orcidlink{0009-0005-9682-3312} \and
Yakir Vizel\,\orcidlink{0000-0002-5655-1667}}
% %
\authorrunning{O. Rappoport et al.}
% % First names are abbreviated in the running head.
% % If there are more than two authors, 'et al.' is used.
% %
\institute{Technion - Israel Institute of Technology, Haifa, Israel\\
\email{\{omer.r,orna,yvizel\}@cs.technion.ac.il}}
%\author{}
%\institute{}
%
\maketitle              % typeset the header of the contribution
\vspace{-15pt}
\input{sections/abstract}
\input{sections/intro_new}
\input{sections/preliminaries}

\input{sections/theory_transformations}
\input{sections/problem_definition}

\input{sections/algorithm}

\input{sections/implementation}
\input{sections/evaluation}
\input{sections/conclusion}
%
%
%
%
% ---- Bibliography ----
%
% BibTeX users should specify bibliography style 'splncs04'.
% References will then be sorted and formatted in the correct style.
%
\bibliographystyle{splncs04}
\bibliography{refs}
\clearpage
\appendix
\input{appendix/proofs}
\input{appendix/benchmarks}

\end{document}

%% file: sections/abstract.tex
\begin{abstract}

Deciding satisfiability of Constrained Horn Clauses (CHCs) modulo the theory of fixed-size bit-vectors ($\theoryBV$) is fundamental to bit-precise program verification.
However, state-of-the-art {\solver}s often struggle with $\theoryBV$, limiting scalability in bit-precise reasoning.

We present \mths, a theory-modular framework for deciding satisfiability of CHCs modulo $\theoryBV$ by combining reasoning in $\theoryBV$ and the theory of Integer Arithmetic ($\theoryIA$).
Given a CHC set modulo $\theoryBV$, \mths partitions it into two fragments interpreted over $\theoryBV$ and $\theoryIA$.
Moreover, it implements an algorithm that reasons about the fragments in a modular fashion, exchanges information between them via sound translations across theories, and determines satisfiability w.r.t. the original CHC set.

We implemented a prototype of \mths using Z3 and \spacer and evaluated it on bit-manipulating benchmarks.
Our evaluation shows that \mths significantly outperforms \spacer on these benchmarks.
\end{abstract}

\begin{comment}
Deciding satisfiability of Constrained Horn Clauses (CHCs) modulo the background theory of Fixed-Size Bit-Vectors ($\theoryBV$) is essential for bit-precise program safety verification. However, state-of-the-art {\solver}s often struggle when $\theoryBV$ is used, making such solvers impractical for bit-precise program safety verification.

In this paper, we introduce \mths, a novel theory-modular framework that allows deciding satisfiability of CHCs modulo $\theoryBV$ by using both $\theoryBV$ and the theory of Integer Arithmetic ($\theoryIA$). \mths allows to fragment a set of CHCs into two sets, one over $\theoryBV$ while the other is over $\theoryIA$. Moreover, it implements an algorithm that can reason about the different fragments and decide satisfiability w.r.t. the original set of clauses modulo $\theoryBV$.

We evaluated a prototype of \mths on a benchmark of bit-manipulating programs. Our evaluation shows that \mths significantly outperforms the state-of-the-art \solver \spacer on our benchmark.
\end{comment}

%% file: sections/intro_new.tex
\section{Introduction}

Constrained Horn Clauses (CHCs) is a fragment of First Order Logic (FOL) with many practical applications in Formal Methods~\cite{DBLP:conf/birthday/BjornerGMR15,DBLP:conf/nfm/RothenbergGVS23,DBLP:conf/esop/ItzhakySV24}. 
In particular, program safety can be reduced to satisfiability of a set of CHCs. Often, the reduction of program safety to CHCs is done w.r.t. the theory of Integer Arithmetic ($\theoryIA$). Since {\solver}s perform well when Integer Arithmetic is used as a background theory, this results in efficient program safety verification.
However, when program operations are required to be represented soundly relative to the low-level bit representation of data, \solver{s} do not perform as efficiently.

Encoding program safety verification w.r.t. bit-precise semantics can be naturally achieved through a reduction to %Constrained Horn Clauses (CHCs)
CHCs modulo the theory of Fixed-Size Bit-Vectors ($\theoryBV$). 
Alternatively, bit-precision can be modeled using $\theoryIA$ despite the fact that the domain is infinite. This is done by representing bit-precise operations and modular arithmetic via multiplication, division, and modulo constraints~\cite{DBLP:conf/vmcai/ZoharIMNNPRBT22}\footnote{Note that~\cite{DBLP:conf/vmcai/ZoharIMNNPRBT22} does not consider CHCs. The primary focus of that work is on solving formulas over $\qfBV$ using a bit-precise encoding  w.r.t. $\qfIA$.}. 
However, neither approach consistently yields an efficient verification procedure. 
This encoding choice presents a fundamental trade-off: $\theoryBV$ formulations often hinder the generalization capabilities of \solver{s}, whereas bit-precise $\theoryIA$ encodings produce complex constraints that are computationally expensive to process, particularly in the presence of bit-wise operations.

% Next I want to say that CHC solvers do not perform well for IA with bit-precise encoding because bit-wise operations generate complex terms, while CHC with BV struggle to generalize. We then need to say that our suggested framework solves this by utilizing the efficiency of CHC solvers wrt Integer Arithmetic, while keeping the inefficient part of bit-wise operations in the BV world.

In this paper, we present \mths, a novel framework designed to efficiently decide the satisfiability of CHCs modulo $\theoryBV$ by leveraging the complementary strengths of both $\theoryBV$ and $\theoryIA$.
Given a set of CHCs $\chcs(\theoryBV)$, \mths first partitions the input into two distinct fragments, $(\BVfragment(\theoryBV), \IAfragment(\theoryIA))$, each interpreted under its respective background theory. 
\mths then employs a dual-reasoning algorithm that processes both sets in a modular fashion to determine the satisfiability of the original $\chcs(\theoryBV)$.
The core intuition behind \mths is to exploit the fact that \solver{s} handle CHCs modulo $\theoryIA$ effectively when limited to arithmetic operations. 
By isolating purely arithmetic constraints in $\IAfragment(\theoryIA)$ while retaining bit-wise operations in $\BVfragment(\theoryBV)$, \mths bridges the aforementioned gap, effectively ``enjoying the best of both worlds'' to achieve superior performance.

%In this paper, we overcome this limitation by introducing \mths, a novel framework that allows deciding satisfiability of CHCs using multiple theories. 
%First, instead of considering a monolithic set of CHCs $\chcs(\theory)$, \mths allows encoding the satisfiability of $\chcs(\theory)$ as a tuple of $k$ distinct fragments $( \chcs_1(\theory_1),\ldots,\chcs_k(\theory_k))$. Namely, each fragment $\chcs_i$ is interpreted w.r.t. some background theory $\theory_i$.
%where $\theory_1 = \theory$ \yv{not sure we need that last part here}. 
%Second, \mths implements an algorithm that reasons about the tuple of fragments, determining the satisfiability of the input set $\chcs(\theory)$.
%Despite using different background theories, \mths determines satisfiability w.r.t. one, pre-determined, background theory.
%\yv{Should we say there are two main parts: the encoding of the problem, which is this fragmentation into $k$ sets, and the algorithmic part which operates on the suggested encoding?}

Deciding the satisfiability of $\chcs(\theoryBV)$ using the decomposed sets $(\BVfragment(\theoryBV), \IAfragment(\theoryIA))$ is a non-trivial task that requires bridging the differences in signatures and semantics of $\theoryBV$ and $\theoryIA$ through sound transformations. To enable modular reasoning, \mths integrates three essential components: 
\begin{enumerate*}[label={(\roman*)}]
    \item a \emph{theory transformer} to translate formulas and structures between $\theoryBV$ and $\theoryIA$;
    \item an \emph{interface} to facilitate communication of reasoning steps between $\BVfragment$ and $\IAfragment$; and 
    \item an \emph{algorithm} that leverages these components to analyze both fragments and decide the satisfiability of the original $\chcs(\theoryBV)$ while ensuring soundness.
\end{enumerate*}

Decomposing $\chcs(\theoryBV)$ into $(\BVfragment,\IAfragment)$ is done in three steps.
%The division of $\chcs(\theoryBV)$ into $(\BVfragment,\IAfragment)$ is done in three steps.
First, the set $\chcs$ is partitioned into two distinct subsets $(\BVfragment,\IAfragment)$ such that their union is $\chcs$. 
While the partitioning can be done arbitrarily, in our prototype implementation it is performed according to the constraints that appear in the body of CHCs in $\chcs$. Specifically, CHCs that include only arithmetic operations are placed in $\IAfragment$, while CHCs that include bit-wise operations are placed in $\BVfragment$. 
Next, $\IAfragment$ is re-encoded w.r.t. $\theoryIA$ using a theory transformer, which translates constraints over $\theoryBV$ to constraints over $\theoryIA$. 
%Since CHCs are formulas, this step requires a transformation of formulas over the signature of $\theoryBV$, to formulas over the signature of $\theoryIA$. 
Lastly, $\BVfragment$ and $\IAfragment$ share uninterpreted predicate symbols. 
For every shared predicate symbol $p$, two copies, $p^\BVshort$ and $p^\IAshort$, are created. 
Then, $p$ is substituted with $p^\BVshort$ and $p^\IAshort$ in $\BVfragment$ and $\IAfragment$, respectively. 
If $p$ is an \emph{interface predicate}, namely, w.l.o.g. $p$ appears in the head of a clause in $\BVfragment$ and in the body of a clause in $\IAfragment$, it is essential to maintain the ``logical connection'' between $\BVfragment$ and $\IAfragment$. 
To achieve that, a set of \emph{interface constraints} $\interfaces$ is generated for all such interface predicates. 
These constraints ensure consistency w.r.t. interpretations of the newly created copies of $p$.
The result of these steps is the tuple $\encoding=(\BVfragment,\IAfragment,\interfaces)$, called a \emph{theory-modular encoding}.

Reasoning over a theory-modular encoding $\encoding=(\BVfragment,\IAfragment,\interfaces)$ is performed by \solve.
The procedure relies on existing {\solver}s for $\theoryBV$ and $\theoryIA$.
It generates single-theory queries, delegates them to the corresponding {\solver}s, and coordinates the exchange of information between fragments as dictated by the interface constraints.
%an algorithm that generalizes the \spacer~\cite{DBLP:journals/fmsd/KomuravelliGC16} algorithm for CHC solving. Recall that \spacer itself is a generalization of Property Directed Reachability (PDR)~\cite{DBLP:conf/vmcai/Bradley11}. Hence, similarly to \spacer, \solve starts from the query and performs a backward search over the structure of the CHCs. In the case of \mths, since the structure is ``fragmented'', \solve needs to move between the different sets of CHCs. To do so, it uses the interface constraints. Given a \emph{proof obligation} (POB) generated in $\chcs_j$, in order to continue the backward traversal through $\chcs_i$, the POB, which is a formula in $\theory_j$ must be translated to a formula in $\theory_i$. This is done using a \emph{theory transformer} and the interface constraints between $\chcs_i$ and $\chcs_j$. Dually, when \solve learns a lemma, which serves as a candidate interpretation for an uninterpreted interface predicate $p^i$ in $\chcs_i$, it uses theory transformation and the interface constraints in order to also add the appropriate lemma to $p^j$ in $\chcs_j$.

We implemented a \emph{proof-of-concept} of \mths using Z3~\cite{DBLP:conf/tacas/MouraB08} and its \solver \spacer~\cite{DBLP:journals/fmsd/KomuravelliGC16}. For the evaluation, we created a benchmark with several examples that highlight the benefits of \mths over a standard \solver. This benchmark set consists of small realistic programs that use bit-manipulating, as well as arithmetic, operations. Our evaluation shows that while \spacer cannot handle these examples, even when considering bit-vectors of a relatively small size, 
%(e.g. bit-vectors of size 8), 
\mths easily solves these instances for much larger bit-vector sizes.% (up to 63).

\input{benchmarks/opposite_signs}

\paragraph{\bf Motivating Example.}
Our motivation comes from program safety verification. Assume that the reduction to CHCs is done modulo $\theoryBV$ in order to capture program safety soundly relative to the low-level bit representation of data. Using $\theoryBV$ as a background theory, however, limits the kind of programs that can be analyzed since {\solver}s often do not perform well when $\theoryBV$ is used. 
%As an alternative, one could use $\theoryIA$ as a background theory. Note that {\solver}s perform well for the theory of {$\theoryIA$}. However, if the program may experience an overflow, consists of non-linear (e.g. multiplication) or bit-precise (e.g. bit-wise XOR) operations, choosing $\theoryIA$ leads to a loss of precision.

Now, consider the set of clauses in Fig.~\ref{fig:opposite_sign}(b). It encodes a simple program (Fig.~\ref{fig:opposite_sign}(a)) that receives an input $x$ greater than 0, and initializes variables $a$ and $b$ to 0. Next it performs $x$ iterations of a loop, incrementing $a$ by 1, while decrementing $b$ by 1 at each iteration. The property states that the bit-wise XOR of $a$ and $b$ is always negative. This bit-wise operation is equivalent to checking whether $a$ and $b$ have opposite signs, which clearly holds for this program. When encoded using $\theoryBV$ (Fig.~\ref{fig:opposite_sign}(b)), given a time limit of 2 hours, \spacer can only solve it for bit-vectors of width 9. When encoding this problem using bit-precise $\theoryIA$~\cite{DBLP:conf/vmcai/ZoharIMNNPRBT22} \spacer can only solve the resulting CHCs for an encoding that captures bit-vectors of size 3, given the same 2 hours time limit.
%The keen-eyed reader may wonder about the performance of \spacer when encoding this problem using bit-precise $\theoryIA$~\cite{DBLP:conf/vmcai/ZoharIMNNPRBT22}. In this case, \spacer can only solve the resulting CHCs for an encoding that captures bit-vectors of size 3.

Next, consider the \mths encoding in Fig.~\ref{fig:opposite_sign}(c). The set of clauses is partitioned into two sets $\BVfragment$ and $\IAfragment$ and the encoding is performed w.r.t. $\theories=(\theoryBV,\theoryIA)$. The partitioning is not chosen arbitrarily. $\BVfragment$ includes a bit-precise operation (bit-wise XOR), while $\IAfragment$ includes only arithmetic. 
%Intuitively, {\solver}s perform better on arithmetic operations when using theories like \lia and \nia, and hence it is preferable to steer a way from $\theoryBV$ for this kind of CHCs. 
As a result, we choose $\theoryBV$ for $\BVfragment$ and $\theoryIA$ for $\IAfragment$. Using this encoding, \mths is able to solve this example instantly for bit-vectors up to a size of 62, considerably outperforming \spacer. 

\subsection{Related Work}

To the best of our knowledge, prior work on CHC-SAT has not explored a theory-modular approach that combines multiple background theories within a single framework. 
The work in~\cite{DBLP:conf/atva/RappoportGV23} partitions a set of CHCs based on its structure, but all subsets are interpreted over the same background theory. 
In contrast, \mths enables coordinated reasoning across different theories.

In program verification, several works combine reasoning over $\theoryBV$ and $\theoryIA$. 
%For example, \cite{DBLP:conf/tacas/GurfinkelBM14} suggests to iteratively reason about a program using $\theoryIA$, and then translates the result to $\theoryBV$ until a valid solution is found.
%Other works, like~\cite{DBLP:conf/tacas/SchusseleBDHJKP24,10.1145/3728905}, rely on abstraction and refinement.
For example,~\cite{DBLP:conf/tacas/GurfinkelBM14} alternates between reasoning in $\theoryIA$ and validating in $\theoryBV$, while other approaches~\cite{DBLP:conf/tacas/SchusseleBDHJKP24,10.1145/3728905} rely on abstraction and refinement. 
In contrast, \mths does not use abstraction and captures the semantics of $\theoryBV$ precisely.

Transforming logical formulas between theories, both syntactically and semantically, is not a new concept. 
In Institution theory~\cite{DBLP:conf/lop/GoguenB83,DBLP:journals/jacm/GoguenB92}, such transformations are formalized via signature and institution morphisms. 
In the context of Satisfiability Modulo Theories (SMT), similar translations have been developed, in particular for $\theoryBV$~\cite{DBLP:conf/tacas/OkudonoK20,DBLP:conf/vmcai/ZoharIMNNPRBT22,DBLP:conf/smt/BarthH24}. 
Theory combination in SMT is classically handled by the Nelson-Oppen framework~\cite{DBLP:journals/toplas/NelsonO79}, which assumes stable infiniteness and thus does not apply to $\theoryBV$, and restricts communication to equalities over shared variables. 
More general frameworks based on polite theories~\cite{DBLP:conf/frocos/RaniseRZ05,DBLP:conf/lpar/JovanovicB10} can handle $\theoryBV$. 
In contrast, \mths operates in the CHC-SAT setting and, rather than combining theories within a single formula, decomposes CHCs into fragments interpreted over different theories and coordinates them via the exchange of general quantifier-free formulas, going beyond equality-based communication.

%% file: benchmarks/opposite_signs.tex
\begin{figure}[t]
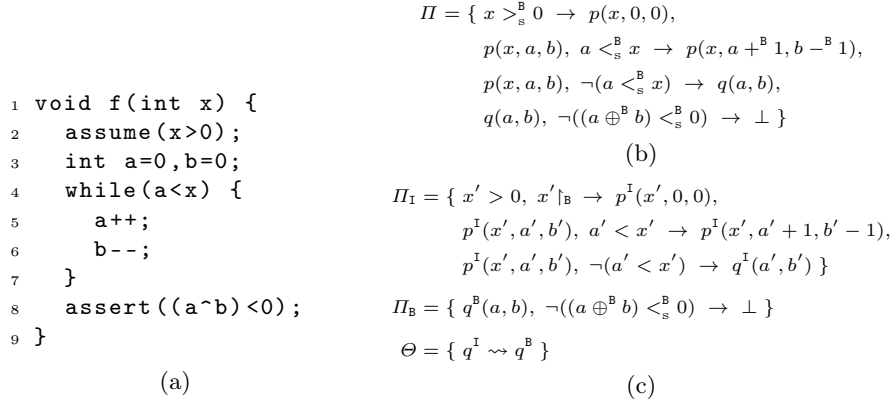

\centering

% Left column: Panel 1
\begin{minipage}[t]{0.32\textwidth}
\vspace{3em} 
\begin{lstlisting}[style=compactC]
void f(int x) {
  assume(x>0);
  int a=0,b=0;
  while(a<x) {
    a++;
    b--;
  }
  assert((a^b)<0);
}
\end{lstlisting}
\centering (a)
\end{minipage}
\hfill
% Right column: Panels 2 and 3 stacked vertically
\begin{minipage}[t]{0.67\textwidth}

% Top: Panel 2
\begin{minipage}[t]{\textwidth}
\[
{\scriptsize
\begin{aligned}
\chcs = \{ \;
& x \bvsgt 0 \; \rightarrow \; p(x,0,0), \\
& p(x,a,b), \; a \bvslt x \; \rightarrow \; p(x,a \bvadd 1,b \bvsub 1), \\
& p(x,a,b), \; \neg(a \bvslt x) \; \rightarrow \; q(a,b), \\
& q(a,b), \; \neg((a \bvxor b) \bvslt 0) \; \rightarrow \; \bot 
\; \}
\end{aligned}
}%
\]
\centering (b)
\end{minipage}

\vspace{0.5em} % small vertical gap

% Bottom: Panel 3
\begin{minipage}[t]{\textwidth}
\[
{\scriptsize 
\begin{aligned}
\chcs_{\IAshort} = \{ \;
  & \varIA{x}>0, \; \restrBV{\varIA{x}} \; \rightarrow \; \upredicateIA{p}(\varIA{x},0,0), \\
  & \upredicateIA{p}(\varIA{x},\varIA{a},\varIA{b}), \; \varIA{a} < \varIA{x} \; \rightarrow \; \upredicateIA{p}(\varIA{x},\varIA{a} + 1,\varIA{b} - 1), \\
  & \upredicateIA{p}(\varIA{x},\varIA{a},\varIA{b}), \; \neg(\varIA{a} < \varIA{x}) \; \rightarrow \; \upredicateIA{q}(\varIA{a},\varIA{b})
\; \} \\[1ex]
\chcs_{\BVshort} = \{ \;
  & \upredicateBV{q}(a,b), \; \neg((a \bvxor b) \bvslt 0) \; \rightarrow \; \bot
\; \} \\[1ex]
\interfaces = \{\; &\interface{\upredicateIA{q}}{\upredicateBV{q}}\; \}
\end{aligned}
}%
\]
\centering (c)
\end{minipage}

\end{minipage}

\caption{The \textsf{opp-signs} benchmark.
(a) A C function.
(b) Its CHC encoding over the bit-vector theory ($\chcs$).
(c) A theory-modular encoding, where arithmetic operations are handled in integer arithmetic ($\chcs_{\IAshort}$), while bit-precise operations remain in the bit-vector fragment ($\chcs_{\BVshort}$), with the original connections preserved via interface constraints ($\interfaces$).
All variables and constants in $\chcs$ have sort $\sortBV{n}$ for some $n \in \mathbb{N}^+$.
Bit-vector constants are written using their signed integer values for readability.
The term $(a \bvxor b) \bvslt 0$ evaluates to true if and only if $a$ and $b$ have opposite signs.}
\label{fig:opposite_sign}
\end{figure}

%% file: sections/preliminaries.tex
\section{Preliminaries}\label{sec:preliminaries}

\subsection{Many-Sorted First-Order Logic with Equality}
We assume the reader is familiar with \emph{first-order logic} and \emph{satisfiability}.

\paragraph{Syntax}
We fix an infinite set $\mathbf{S}$ of \emph{sort symbols} and an infinite set $\variables$ of \emph{variables}, where each variable is uniquely associated with a sort in $\mathbf{S}$.
A \emph{signature} $\signature$ is a tuple $(\sorts,\predicates,\functions)$, where $\sorts \subseteq \mathbf{S}$ is a set of sort symbols, and $\predicates$ and $\functions$ are sets of \emph{predicate} and \emph{function symbols}, respectively.
Each predicate symbol $p \in \predicates$ (resp., function symbol $f \in \functions$) is assigned a unique \emph{arity} $n \in \mathbb{N}$ and a \emph{type} $\predtype$ (resp., $\functype$), where $\sort_i, \sort \in \sorts$.
Function symbols of arity $0$ are called \emph{constant symbols}.
$\signature$-\emph{terms}, $\signature$-\emph{atoms}, and $\signature$-\emph{formulas} are defined in the standard way, subject to the requirement that they are \emph{well-sorted}, i.e., all function and predicate applications must respect their types.
A $\signature$-\emph{sentence} is a formula with no free variables.
We write $\varphi[\bar{x}]$ to denote a formula whose free variables are in $\bar{x}$.
The constants $\top$ and $\bot$ denote the truth values \texttt{true} and \texttt{false}, respectively.

\paragraph{Semantics} 
Given a signature $\signature = (\sorts, \predicates, \functions)$, a $\signature$-\emph{structure} $\structure$ maps each sort $\sort \in \sorts$ to a non-empty set $\domain_\sort^\structure$, called the \emph{domain} of $\sort$ in $\structure$.
It also maps each predicate symbol $p \in \predicates$ of type $\sort_1 \times \cdots \times \sort_n$ to a relation $p^\structure \subseteq \domain_p^\structure$, where $\domain_p^\structure = \domain_{\sort_1}^\structure \times \cdots \times \domain_{\sort_n}^\structure$, and similarly for function symbols.
For a variable $x \in \variables$ of sort $\sort$, we write $\domain_x^\structure$ for $\domain_\sort^\structure$, and extend this notation to tuples $\bar{x}$ in the natural way.
When a formula has free variables, satisfiability is defined relative to a structure and an \emph{assignment}, i.e., a partial function mapping variables to elements of their domains.
For $\bar{x} = (x_1, \dots, x_n)$ and $\bar{a} = (a_1, \dots, a_n) \in \domain_{\bar{x}}^\structure$, we write $\assignment{\bar{x}}{\bar{a}}$ for the assignment that maps each $x_i$ to $a_i$.
The \emph{satisfiability relation} is defined in the standard way.
We write $\structure, \assignment{\bar{x}}{\bar{a}} \models \varphi[\bar{x}]$ to denote that $\structure$ and $\assignment{\bar{x}}{\bar{a}}$ satisfy $\varphi$.
If $\varphi$ is a sentence, we simply write $\structure \models \varphi$.
%, and say that $\structure$ is a \emph{model} of $\varphi$.
For a $\signature$-formula $\varphi$ with free variables $\bar{x}$, the \emph{semantics} of $\varphi$ with respect to $\structure$, denoted $\semantics{\varphi}{\structure}{}$, is the set of all tuples that satisfy $\varphi$ under $\structure$, that is, $\semantics{\varphi}{\structure}{}=\set{\bar{a}\in\domain_{\bar{x}}^\structure}{\structure,\assignment{\bar{x}}{\bar{a}}\models\varphi}$.
We sometimes extend a signature with fresh predicate symbols.
Let $\predicates'$ be disjoint from $\predicates$, with types over $\sorts$. 
We write $\signature \cup \predicates'$ for the signature $(\sorts, \predicates \cup \predicates', \functions)$. 
For each $p' \in \predicates'$, let $R_{p'} \subseteq \domain_{p'}^\structure$. 
We write $\structure\{p' \mapsto R_{p'}\}_{p' \in \predicates'}$ for the $(\signature \cup \predicates')$-structure that interprets each $p'$ as $R_{p'}$ and agrees with $\structure$ on all symbols of $\signature$.

\paragraph{Theories} 
Theories are used to constrain the interpretation of certain symbols in a signature.
A $\signature$-\emph{theory} is a class of $\signature$-structures, called its \emph{models}.
A structure $\structure$ \emph{satisfies} a $\signature$-sentence $\varphi$ \emph{modulo} $\theory$, written $\structure \models_{\theory} \varphi$, if $\structure$ is a $\theory$-model and $\structure \models \varphi$.
A sentence $\varphi$ is \emph{satisfiable modulo $\theory$} if some $\theory$-model satisfies it, and \emph{valid modulo $\theory$}, written $\models_{\theory} \varphi$, if all $\theory$-models satisfy it.
Next, we introduce two theories, which are the focus of this work.

The signature $\signatureBV = (\sorts_{\BVshort}, \predicates_{\BVshort}, \functions_{\BVshort})$ and the theory $\theoryBV$ of fixed-size bit-vectors are defined in SMT-LIB~2~\cite{barrett2010smt}.\footnote{We use a subset sufficient to express all bit-vector symbols.}
The set $\sorts_{\BVshort}$ contains all sorts $\sortBV{n}$ for $n \in \mathbb{N}^+$, where $n$ denotes the bit-width.
The set $\predicates_{\BVshort}$ includes, for each $n$, equality $=$, unsigned comparisons $\bvult,\bvule$, and signed comparisons $\bvslt,\bvsle$ of type $\sortBV{n} \times \sortBV{n}$.\footnote{Symbols in $\signatureBV$ are overloaded for each $\sortBV{n}$.}
The set $\functions_{\BVshort}$ includes, for each $n \in \mathbb{N}^+$, a constant symbol of type $\sortBV{n}$ for each bit-vector of width $n$, represented as a bit-string.\footnote{We sometimes refer to a bit-vector constant by its unsigned or signed integer value.}
It also includes arithmetic operators $\bvadd,\bvsub,\bvmul,\bvudiv,\bvurem$, bitwise operators $\bvnot,\bvand,\bvor,\bvxor$, and shift operators $\bvshl,\bvlshr$, all of type $\sortBV{n} \times \sortBV{n} \to \sortBV{n}$ (except $\bvnot : \sortBV{n} \to \sortBV{n}$), as well as extraction $\bvextract{u}{l} : \sortBV{n} \to \sortBV{u-l+1}$ for $0 \le l \le u < n$, and concatenation $\bvconcat : \sortBV{n} \times \sortBV{m} \to \sortBV{n+m}$.
The class of quantifier-free $\signatureBV$-formulas is denoted $\qfBV$.
Throughout the paper, we fix a $\theoryBV$-model $\modelBV$

The signature $\signatureIA = (\sorts_{\IAshort}, \predicates_{\IAshort}, \functions_{\IAshort})$ and the theory $\theoryIA$ of integer arithmetic are also defined in SMT-LIB~2.
The set $\sorts_{\IAshort}$ consists of the single sort $\sortIA$.
The set $\predicates_{\IAshort}$ includes equality $=$ and comparisons $<,\leq$ of type $\sortIA \times \sortIA$.
The set $\functions_{\IAshort}$ includes a constant symbol of type $\sortIA$ for each integer.
It also includes the arithmetic operators $\intadd,\intsub,\intmul,\intdiv,\intrem$ of type $\sortIA \times \sortIA \to \sortIA$.
The class of quantifier-free $\signatureIA$-formulas is denoted $\qfIA$.
Throughout the paper, we fix a $\theoryIA$-model $\modelIA$.

\subsection{Constrained Horn Clauses}\label{subsec:chcs}

\paragraph{Syntax}
Let $\signature = (\sorts, \predicates, \functions)$ be a signature, and let $\upredicates$ be a set of fresh predicate symbols, called \emph{uninterpreted predicates}, disjoint from $\predicates$ and typed over $\sorts$. 

A \emph{Constrained Horn Clause} (\emph{CHC}) over $\signature \cup \upredicates$ is a sentence of the form $\forall \bar{x}.\; \body[\bar{x}] \rightarrow \head[\bar{x}]$, for which:
\begin{enumerate*}[label=(\roman*)]
    \item The \emph{body} $\body[\bar{x}]$ is $\phi[\bar{x}] \wedge \bigwedge_i p_i(\bar{x}_{p_i})$, where $\phi$ (the \emph{constraint}) is a quantifier-free $\signature$-formula over $\bar{x}$\footnote{For simplicity, we assume all variables in $\bar{x}$ appear in the constraint.}, and each $p_i(\bar{x}_{p_i})$ is an application of $p_i \in \upredicates$ with $\bar{x}_{p_i} \subseteq \bar{x}$; and
    \item The \emph{head} $\head[\bar{x}]$ is either $h(\bar{x}_h)$ for some $h \in \upredicates$ with $\bar{x}_h \subseteq \bar{x}$, or $\bot$.
\end{enumerate*}

An uninterpreted predicate may appear multiple times in the body with different arguments.
A CHC is a \emph{query} if its head is $\bot$, and a \emph{rule} otherwise. 
A rule is a \emph{fact} if its body contains no uninterpreted predicates.

A set of CHCs can encode a safety verification problem for a sequential program.
In such an encoding, the uninterpreted predicates serve as placeholders for inductive invariants at specific program locations, while the constraints capture the program’s transitions between these locations.
Facts encode assumptions on initial states, and queries dually encode assertions.

\paragraph{Semantics}
Fix a set of CHCs $\chcs$ over $\signature \cup \upredicates$ and a $\signature$-theory $\theory$.
To define satisfiability of $\chcs$ modulo $\theory$, we first define the \emph{semantics} of uninterpreted predicates.
Given a $\theory$-model $\model$, an uninterpreted predicate $p \in \upredicates$, and $n \in \mathbb{N}$, the \emph{bounded semantics} of $p$ w.r.t.\ $\model$ at bound $n$, denoted $\semantics{p}{\model}{n}$, is the set of tuples in $\domain_p^\model$ derivable for $p$ from $\chcs$ by derivations of depth at most $n$.
In the context of program verification, $\semantics{p}{\model}{n}$ corresponds to the set of program states reachable at the location of $p$ by executions inducing derivations of depth at most $n$.
The bounded semantics is defined inductively on $n$ as follows.

\begin{definition}[Bounded Semantics]\label{def:bounded_semantics}
    Let $\model$ be a $\theory$-model, $p \in \upredicates$ an uninterpreted predicate symbol, and $n \in \mathbb{N}$ a bound.
    The \emph{bounded semantics} of $p$ with respect to $\model$ at bound $n$, denoted $\semantics{p}{\model}{n}$, is defined as follows:
    {%\small
    \begin{align}
        \semantics{p}{\model}{n} =& 
        \bigcup_{(\forall\bar{x}. \; \body[\bar{x}] \rightarrow p(\bar{x}_p)) \in \chcs}
        {\semantics{\exists (\bar{x} \setminus \bar{x}_p). \; \body[\bar{x}]}{\model'}{}} \label{eq:bounded_semantics} \\
        & \text{where }  \model'=
        \begin{cases}
            \model\{q\mapsto \emptyset\}_{q\in\upredicates} & \textup{if } n = 0 \\ \vspace*{-0.25cm}\\
            \model\{q\mapsto \semantics{q}{\structure}{n-1}\}_{q\in\upredicates} & \textup{if } n > 0
        \end{cases}
        \nonumber
    \end{align}
    }%
\end{definition}

%\begin{example}\label{example:bounded_sematics_single_theory}
%    Consider the CHC set $\chcs$ in Fig.~\ref{fig:opposite_sign}(b), where all variables have sort $\sortBV{n}$ for $n \geq 2$.
%    Then $(1,0,0) \in \semantics{p}{\modelBV}{0}$ (derived from the first CHC), $(1,1,-1) \in \semantics{p}{\modelBV}{1}$ (from the second CHC), and $(1,-1) \in \semantics{q}{\modelBV}{2}$ (from the third CHC).
%\end{example}

\begin{example}\label{example:bounded_semantics_single_theory}
    Consider the CHC set $\chcs$ in Fig.~\ref{fig:opposite_sign}(b), where all variables have sort $\sortBV{2}$. 
    Here we use binary bit-vector notation. %(e.g., $01=1$, $11=-1$). % under signed semantics).
    Then $(01,00,00) \in \semantics{p}{\modelBV}{0}$ (derived from the first CHC), $(01,01,11) \in \semantics{p}{\modelBV}{1}$ (from the second CHC), and $(01,11) \in \semantics{q}{\modelBV}{2}$ (from the third CHC).
\end{example}

Observe that the bounded semantics are defined in a monotonic manner: for all $n \in \mathbb{N}$, $\semantics{p}{\model}{n} \subseteq \semantics{p}{\model}{n+1}$.
The bounded semantics of all uninterpreted predicates in $\chcs$ stabilize at the first limit ordinal $\omega$ yielding their (unbounded) semantics~\cite{DBLP:conf/birthday/BjornerGMR15}.
The \emph{semantics} of $p$ with respect to $\model$, denoted $\semantics{p}{\model}{}$, is given by $\bigcup_{n \in \mathbb{N}} \semantics{p}{\model}{n}$.

The satisfiability of a CHC set can be defined via the semantics of its uninterpreted predicates.
Every $\theory$-model $\model$ can be expanded to a $(\signature \cup \upredicates)$-structure $\model\{p \mapsto \semantics{p}{\model}{}\}_{p \in \upredicates}$, which interprets each $p$ by its semantics w.r.t.\ $\model$.
By Definition~\ref{def:bounded_semantics}, this expansion satisfies all rules in $\chcs$.
We say that $\chcs$ is \emph{satisfiable modulo $\theory$} if there exists a $\theory$-model whose expansion also satisfies all \emph{queries}.

\begin{definition}[Satisfiability]\label{def:satisfiability}
$\chcs$ is \emph{satisfiable modulo $\theory$} if there exists a $\theory$-model $\model$ such that $\model\{p \mapsto \semantics{p}{\model}{}\}_{p \in \upredicates} \models \chc$ for every query $\chc \in \chcs$.
\end{definition}

\begin{definition}[The CHC-SAT Problem]
    Given a CHC set $\chcs$ over $\signature \cup \upredicates$ and a $\signature$-theory $\theory$, determine whether $\chcs$ is satisfiable modulo $\theory$.
\end{definition}

When $\chcs$ is unsatisfiable modulo $\theory$, the bounded semantics can be used to construct a derivation of $\bot$~\cite{DBLP:conf/sas/BjornerMR13}, denoted $\chcs \vdash_{\theory} \bot$. 
Such a derivation is called a \emph{refutation}. 
In the context of safety verification, a refutation corresponds to a finite execution leading to an assertion violation.

%If $\chcs$ is unsatisfiable modulo $\theory$, then for every $\theory$-model $\model$ there exists a refutation, i.e., a ground derivation of $\bot$~\cite{DBLP:conf/sas/BjornerMR13}, induced by the bounded semantics of the uninterpreted predicates. %denoted $\chcs \vdash_{\theory} \bot$ + root + child + leaf
%If $\chcs$ is unsatisfiable modulo $\theory$, then for every $\theory$-model $\model$, $\chcs$ admits a refutation, i.e., a ground derivation of $\bot$, denoted $\chcs \vdash_{\theory} \bot$~\cite{DBLP:conf/sas/BjornerMR13}.
%A refutation $\refutation$ has the form of a tree, whose root is labeled by $\bot$.
%Each other node in $\refutation$ at depth $n$ is labeled by an uninterpreted predicate $p$ and a tuple in $\semantics{p}{\model}{n}$. 
%The refutation is defined inductively as follows.
%Leaves are labeled by tuples in the bounded semantics at bound~0 (i.e., tuples derivable from facts in $\chcs$), and each internal node at depth $i$ is labeled by a tuple in the bounded semantics at bound~$i$.
%, obtained by instantiating a rule of $\chcs$ with the tuples labeling its children.
%In safety verification, such a derivation corresponds to a finite execution leading to an assertion violation.

In practice, \textsc{CHC-SAT} algorithms search for interpretations of the uninterpreted predicates that over-approximate their semantics while satisfying the queries.
In safety verification, such interpretations correspond to inductive invariants sufficient to prove safety.
The search is typically restricted to interpretations representable by formulas~\cite{DBLP:conf/birthday/BjornerGMR15}.
A \emph{$\signature$-interpretation} $\interpretation$ assigns to each $p \in \upredicates$ of type $\predtype$ a $\signature$-formula $\interpretation(p)$ with free variables $\bar{y}_p = (y_1, \dots, y_n)$, with $y_i$ of sort $\sigma_i$.
For $\chc \in \chcs$, we write $\interpretation(\chc)$ for the $\signature$-formula obtained by replacing each uninterpreted predicate application $p(\bar{x}_p)$ that occurs in $\chc$ with $\interpretation(p)[\bar{y}_p \mapsto \bar{x}_p]$.

\begin{definition}[Solution]\label{def:solution}
A $\signature$-interpretation $\interpretation$ \emph{satisfies} $\chcs$ modulo $\theory$, denoted $\interpretation \models_{\theory} \chcs$, if $\interpretation(\chc)$ is valid modulo $\theory$ for every $\chc \in \chcs$.
In this case, $\interpretation$ is called a \emph{$\signature$-solution} for $\chcs$.
\end{definition}

If $\interpretation$ is a $\signature$-solution, then for every $\theory$-model $\model$, the expansion $\model\{p \mapsto \semantics{\interpretation(p)}{\model}{}\}_{p \in \upredicates}$ satisfies all $\chc \in \chcs$.
Moreover, $\semantics{p}{\model}{} \subseteq \semantics{\interpretation(p)}{\model}{}$ for all $p \in \upredicates$.
Thus, if a $\signature$-solution exists, then $\chcs$ is satisfiable.
The converse, however, does not necessarily hold due to the limited expressiveness of the class of $\signature$-formulas.

%% file: sections/theory_transformations.tex
\section{Theory Transformers}\label{sec:transformers}

As a first step toward our theory-modular approach to deciding the satisfiability of CHCs modulo $\theoryBV$, we introduce a pair of translation functions, called \emph{theory transformers}.
The first transformer, denoted $\formulaBVIA$, translates bit-vector formulas into integer arithmetic, while the second transformer, denoted $\formulaIABV$, translates formulas in the opposite direction.
These transformers form the basis of our approach.
Given a set $\chcs$ of CHCs over $\signatureBV$, the transformer $\formulaBVIA$ is used to translate a selected fragment of $\chcs$ to $\signatureIA$.
Later, when deciding satisfiability of $\chcs$, each fragment is reasoned about in its respective theory, and the transformers $\formulaBVIA$ and $\formulaIABV$ enable the transfer of information between fragments.
Crucially, this process preserves the semantics of $\chcs$, even though the fragments are defined over different signatures and theories and involve variables of different sorts.

For the remainder of this section, let $\chcs$ be a set of CHCs over the signature $\signatureBV \cup \upredicates$, where $\upredicates$ is a set of uninterpreted predicate symbols.
%Throughout the paper, we fix a $\theoryBV$-model $\modelBV$ and a $\theoryIA$-model $\modelIA$.
We fix a bijective variable mapping $\varBVIA$ between bit-vector variables (i.e., variables of sort $\sortBV{n}$ for some $n \in \mathbb{N}^+$) and integer variables (variables of sort $\sortIA$).
To simplify the notation, for every bit-vector variable $x$ we write $x'$ to denote the unique integer variable such that $\varBVIA(x)=x'$.
For every tuple $\bar{x}=(x_1,\dots,x_m)$ of bit-vector variables, we write $\bar{x}'$ to denote $(x'_1,\dots,x'_m)$.

To preserve semantics across theories, for each bit-vector variable $x$ of sort $\sortBV{n}$ (for some $n \in \mathbb{N}^+$), we relate the domain of $x$ in $\modelBV$ to the domain of its corresponding integer variable $x'$ in $\modelIA$.
The domain $\domBVvar{x} = \domBV{n}$ is finite and consists of all bit-vectors of width $n$, whereas $\domIAvar{x'} = \domIA$ is infinite and includes the integers.
This correspondence is captured by a \emph{domain transformer}.%, defined as follows.

\begin{definition}[Domain Transformer]\label{def:domain_transformers}
	Let $x$ be a variable of sort $\sortBV{n}$ for $n \in \mathbb{N}^+$. 
    %, and let $x' = \varBVIA(x)$.
	A \emph{domain transformer} for $x$ is an injective mapping
	$\domBVIA{x} \colon \domBVvar{x} \to \domIAvar{x'}$.
\end{definition}

To ensure consistency, the only restriction on domain transformers is that all variables appearing in the same argument position of an uninterpreted predicate application in $\chcs$ must be assigned the same mapping.\footnote{In CHCs, since all variables are universally quantified, the same variable may appear in different contexts. To allow greater flexibility in the choice of domain transformers, one may rename variables so that different CHCs use distinct variables.}
Apart from this restriction, domain transformers can, in principle, be defined arbitrarily.
However, two natural choices that simplify the translation and yield a more natural encoding in integer arithmetic are to use either unsigned or signed semantics for bit-vectors.
Let $x$ be a bit-vector variable of sort $\sortBV{n}$ for some $n \in \mathbb{N}^+$.
The \emph{unsigned domain transformer} of $x$, whose image is $[0, 2^n - 1] \subset \domIAvar{x'}$, maps each $v \in \domBVvar{x}$ to $[v]_{\mathbb{N}} = \sum_{i=0}^{n-1} v[i]\cdot 2^i$.
The \emph{signed domain transformer} of $x$, whose image is $[-2^{n-1}, 2^{n-1}-1] \subset \domIAvar{x'}$, maps each $v \in \domBVvar{x}$ to $[v]_{\mathbb{Z}} = -(v[n-1])\cdot 2^{n-1} + \sum_{i=0}^{n-2} v[i]\cdot 2^i$.\footnote{Signed values are represented using two’s complement.}
From now on, the presentation remains agnostic to the concrete choice of domain transformers.
The specific instantiation used in the implementation is discussed in Section~\ref{sec:implementation}.

\begin{comment}
In programming languages, the sign of a variable is typically part of its type, for example \texttt{int} versus \texttt{unsigned int} in C.
In contrast, in $\signatureBV$ the sort of a variable depends only on its bit-width, while signedness affects certain predicate and function symbols, such as $\bvult$ versus $\bvslt$.
As a result, the same bit-vector variable may appear in both unsigned and signed predicate or function applications.
We therefore adopt the following strategy.
For each bit-vector variable $x$, if $x$ participates in more unsigned than signed predicate and function applications in $\chcs$, then $\domBVIA{x}(v) = [v]_{\mathbb{N}}$ for all $v \in \domBVvar{x}$; otherwise, $\domBVIA{x}(v) = [v]_{\mathbb{Z}}$.
For clarity of exposition, however, we restrict attention to a unified strategy in which all domain transformers are defined consistently.
In particular, if unsigned predicate and function applications are more prevalent in $\chcs$ than signed ones, then all domain transformers are defined using unsigned interpretations of bit-vectors; otherwise, all domain transformers are defined using signed interpretations.\yv{You may want to keep this ``more in signed/unsigned'' discussion to implementation details. Then, you could explain here that it can either be this or that, and say that for presentation you assume...}
\end{comment}

Let $\bar{x} = (x_1,\dots,x_m)$ be a tuple of bit-vector variables.
For a tuple $\bar{v} = (v_1,\dots,v_m) \in \domBVvar{\bar{x}} = \domBVvar{x_1} \times \cdots \times \domBVvar{x_m}$, we write $\domBVIA{\bar{x}}(\bar{v})$ to denote the tuple $(\domBVIA{x_1}(v_1),\dots,\domBVIA{x_m}(v_m))$.
By abuse of notation, we extend $\domBVIA{\bar{x}}$ element-wise to subsets of $\domBVvar{\bar{x}}$.
For the inverse direction, 
%let $x' = \varBVIA(x)$.
for any integer value $v' \in \domIAvar{x'}$ that lies in the image of $\domBVIA{x}$, we write $\domIABV{x}(v')$ to denote the unique bit-vector value $v \in \domBVvar{x}$ such that $\domBVIA{x}(v) = v'$.
We lift the inverse domain transformer to tuples and sets analogously.
For an uninterpreted predicate $p \in \upredicates$ that occurs in $\chcs$, we use $\domBVIA{p}$ as shorthand for $\domBVIA{\bar{x}_p}$, where $\bar{x}_p$ is a tuple of variables such that the application $p(\bar{x}_p)$ appears in some CHC in $\chcs$.
This notation is well defined due to the consistency requirement, which ensures that variables in the same argument position of an uninterpreted predicate application in $\chcs$ are assigned the same domain transformer.

%corresponding to an arbitrary application of $p$ in $\chcs$ (the choice is irrelevant due to the consistency requirement).

Let $\varphi'$ be a quantifier-free $\signatureIA$-formula with free variables $\bar{x}'$.
Recall that $\semantics{\varphi'}{\modelIA}{}$ denotes the set of tuples satisfying $\varphi'$ under $\modelIA$.
We sometimes restrict $\semantics{\varphi'}{\modelIA}{}$ to elements corresponding to bit-vector values in $\modelBV$ via the domain transformers.
To this end, we define the \emph{restriction} of $\semantics{\varphi'}{\modelIA}{}$ to $\modelBV$.

\begin{definition}[Domain Restriction]\label{def:domain_restriction}
    Let $\varphi' \in \qfIA$ with free variables $\bar{x}'$. %, and let $\bar{x} = \varIABV(\bar{x}')$.
    The \emph{restriction} of $\semantics{\varphi'}{\modelIA}{}$ to $\modelBV$, denoted $\restrBV{\semantics{\varphi'}{\modelIA}{}}$, is defined as follows:
    \begin{displaymath}
        \restrBV{\semantics{\varphi'}{\modelIA}{}}
        \;=\;
        \semantics{\varphi'}{\modelIA}{} \;\cap\; \domBVIA{\bar{x}}\bigl(\domain_{\bar{x}}^{\modelBV}\bigr)
    \end{displaymath}
\end{definition}

\begin{example}\label{example:semantic_restriction}
    Let $(x,y)$ be a pair of bit-vector variables of sort $\sortBV{2}$, and assume that $\domBVIA{x}$ and $\domBVIA{y}$ are the unsigned domain transformers.
    \begin{align*}
        &\restrBV{\semantics{x' < y'}{\modelIA}{}}
        \;=\;
        \semantics{x' < y'}{\modelIA}{} \;\cap\; \domBVIA{(x,y)}\bigl(\domBVvar{(x,y)}\bigr) 
        \;=\; \\
        & \{(a,b) \in (\domIA)^2 \;|\; a<b\} \cap \{0,1,2,3\}^2 
        =
        \{ (0,1), (0,2), (0,3), (1,2), (1,3), (2,3) \}
    \end{align*}
\end{example}

\begin{comment}
\begin{example}\label{example:semantic_restriction}
    Let $(x,y)$ be a pair of bit-vector variables of sort $\sortBV{2}$, and assume that $\domBVIA{x}$ and $\domBVIA{y}$ are the unsigned domain transformers.
    \begin{align*}
        &\restrBV{\semantics{x' = y' + 1}{\modelIA}{}}
        \;=\;
        %\semantics{x' = y' + 1}{\modelIA}{} \;\cap\; \domBVIA{}\bigl(\domain_{(x,y)}^{\modelBV}\bigr) \\
        \semantics{x' = y' + 1}{\modelIA}{} \;\cap\; \domBVIA{(x,y)}\bigl(\domBVvar{(x,y)}\bigr) 
        \;=\; \\
        &\set{(a+1,a)}{a \in \domIA} \cap (\{0,1,2,3\} \times \{0,1,2,3\}) 
        \;=\;
        \{ (1,0), (2,1), (3,2) \}
    \end{align*}
\end{example}    
\end{comment}

Restrictions can also be performed syntactically by bounding variable domains.
Let $x$ be a bit-vector variable of sort $\sortBV{n}$ for some $n \in \mathbb{N}^+$. %, and let $x' = \varBVIA(x)$.
We write $\restrBV{x'}$ to denote the $\signatureIA$-constraint enforcing that $x'$ ranges over the image of $\domBVIA{x}$, that is, $0 \leq x' \wedge x' \leq 2^n - 1$ for the unsigned domain transformer, and $-2^{n-1} \leq x' \wedge x' \leq 2^{n-1} - 1$ for the signed domain transformer.
%This notion is lifted to formulas as follows.
%Let $\varphi' \in \qfIA$ with free variables $\bar{x}'$.
%The \emph{restriction} of $\varphi'$ to $\modelBV$, denoted $\restrBV{\varphi'}$, is the formula $\varphi' \;\wedge\; \bigwedge_{x' \in \bar{x}'} \restrBV{x'}$.

\begin{comment}
\begin{definition}[Syntactic Domain Restriction]\label{def:syntactic_domain_restriction}
	Let $\varphi' \in \qfIA$ with free variables $\bar{x}'$.
	The \emph{restriction} of $\varphi'$ to $\modelBV$, denoted $\restrBV{\varphi'}$, is defined by
	\begin{displaymath}
		\restrBV{\varphi'}
		\;=\;
		\varphi' \;\wedge\; \bigwedge_{x' \in \bar{x}'} \restrBV{x'}
	\end{displaymath}
\end{definition}    
\end{comment}

%The following lemma formalizes the connection between semantic and syntactic domain restrictions.
%Definitions~\ref{def:domain_restriction} and~\ref{def:syntactic_domain_restriction}.

\begin{lemma}\label{lemma:syntactic_restriction}
	Let $\varphi' \in \qfIA$ with free variables $\bar{x}'$. 
    The following holds: 
    \begin{displaymath}
        \restrBV{\semantics{\varphi'}{\modelIA}{}} = \semantics{\varphi' \;\wedge\; \bigwedge_{x' \in \bar{x}'} \restrBV{x'}}{\modelIA}{}
    \end{displaymath}
\end{lemma}

\begin{example}
	Following Example~\ref{example:semantic_restriction}, the formula $x' < y' \wedge \restrBV{x'} \! \wedge \,\restrBV{y'}$ is equivalent to $x' < y' \;\wedge\; (0 \leq x' \wedge x' \leq 3) \;\wedge\; (0 \leq y' \wedge y' \leq 3)$.
\end{example}

The restriction operator allows us to compare bit-vector formulas with integer arithmetic formulas by restricting integer-valued assignments to the domain induced by bit-vector values.
Building on this notion, we now formalize what it means for a syntactic translation between bit-vector and integer arithmetic formulas to preserve semantics.
Intuitively, an integer arithmetic formula is considered semantics-preserving with respect to a bit-vector formula if the two formulas agree on this restricted domain.
Note that the integer arithmetic formula may admit additional satisfying assignments outside this domain.

\begin{definition}[Theory Transformers]\label{def:theory_transformers}
	A \emph{theory transformer from $\theoryBV$ to $\theoryIA$} is a translation function $\formulaBVIA$ that translates every formula $\varphi \in \qfBV$ with free variables	$\bar{x}$ into a $\qfIA$ formula $\formulaBVIA(\varphi)$ with free variables $\bar{x}'$ %$\varBVIA(\bar{x})$ 
    such that:
	\begin{equation}\label{eq:theory_transformer_BV_IA}
		\semantics{\varphi}{\modelBV}{}
		\;=\;
		\domIABV{\bar{x}}\!\left(
			\restrBV{\semantics{\formulaBVIA(\varphi)}{\modelIA}{}}
		\right)
	\end{equation}

	Similarly, a \emph{theory transformer from $\theoryIA$ to $\theoryBV$} is a translation function $\formulaIABV$ that translates every formula $\varphi' \in \qfIA$ with free variables $\bar{x}'$ into a $\qfBV$ formula $\formulaIABV(\varphi')$ with free variables $\bar{x}$ %$\varIABV(\bar{x}')$ 
    such that:
	\begin{equation}\label{eq:theory_transformer_IA_BV}
		\semantics{\formulaIABV(\varphi')}{\modelBV}{}
		\;=\;
		\domIABV{\bar{x}}\!\left(
			\restrBV{\semantics{\varphi'}{\modelIA}{}}
		\right)
	\end{equation}
\end{definition}

\begin{example}
    Following Example~\ref{example:semantic_restriction}, the $\signatureIA$-formula $x' < y'$ preserves the semantics of the $\signatureBV$-formula $x \bvult y$ in the sense of Eqs.~\ref{eq:theory_transformer_BV_IA} and~\ref{eq:theory_transformer_IA_BV}.
    Notably, $\semantics{x' < y'}{\modelIA}{}$ is infinite, whereas $\semantics{x \bvult y}{\modelBV}{}$ is finite. The two formulas coincide on the restricted bit-vector domain up to the domain transformer.
\end{example}

\begin{comment}
We now define the concrete theory transformers used in the remainder of the paper.
These transformers instantiate the notion of theory transformers introduced above by fixing how bit-vector formulas are translated into integer arithmetic formulas, and vice versa.

For $\formulaBVIA$, we adopt the translation of Zohar et al.~\cite{DBLP:conf/vmcai/ZoharIMNNPRBT22}.
This translation satisfies the requirement of Eq.~\ref{eq:theory_transformer_BV_IA}.
In fact, it satisfies a stronger requirement as    
\end{comment}

%% file: sections/problem_definition.tex
\section{Theory-Modular Encoding of Bit-Precise CHCs}\label{sec:problem_definition}

Throughout this section and the next, we fix the following setup.
Let $\chcs$ be a set of CHCs over the signature $\signatureBV \cup \upredicates$, where $\upredicates$ is a set of uninterpreted predicate symbols.
For each bit-vector variable $x$, let $\domBVIA{x}$ be a corresponding domain transformer.
%We assume that these domain transformers satisfy the consistency requirement introduced in the previous section, namely that variables occurring in the same argument position of an uninterpreted predicate application are mapped consistently.
Finally, let $\formulaBVIA$ be a theory transformer from $\theoryBV$ to $\theoryIA$, and let $\formulaIABV$ be a theory transformer from $\theoryIA$ to $\theoryBV$.

\subsection{The Encoding}

We now introduce our \emph{theory-modular encoding} of $\chcs$.
The encoding is defined with respect to a partition of $\chcs$ into two fragments.
The first fragment, the \emph{bit-vector fragment}, consists of CHCs reasoned about with respect to $\theoryBV$.
These CHCs are left unchanged, except that each uninterpreted predicate $p$ is replaced by a fresh copy $\upredicateBV{p}$.
The second fragment, the \emph{integer fragment}, consists of CHCs reasoned about with respect to $\theoryIA$.
These CHCs are translated into integer arithmetic formulas using the theory transformer $\formulaBVIA$.
Concretely, the constraint of each CHC is translated into a $\signatureIA$-formula via $\formulaBVIA$, each uninterpreted predicate $p$ is replaced by a fresh copy $\upredicateIA{p}$, and variables are consistently renamed according to the variable mapping.
In addition, the (syntactic) domain restriction operator is applied to every variable not occurring in a predicate application in the body, thereby eliminating integer derivations outside the bit-vector domain.
The encoding also introduces \emph{interface constraints}, which preserve the original dependencies among uninterpreted predicates by relating their copies.
The complete encoding is formalized in the following definition.

\begin{definition}[Theory-Modular Encoding]\label{def:theory_modular_encoding}
    Let $\Delta_{\BVshort} \subseteq \chcs$ and $\Delta_{\IAshort}=\chcs \setminus \Delta_{\BVshort}$.\\
    The \emph{theory-modular encoding} of $\chcs$ induced by $(\Delta_{\BVshort},\Delta_{\IAshort})$ is a triple $\encoding = (\BVfragment,\IAfragment,\interfaces)$ defined as follows:
	\begin{itemize}
		\item $\BVfragment$, the \emph{bit-vector fragment}, is the set of CHCs over the signature $\signatureBV \cup \upredicatesBV$, where $\upredicatesBV = \set{\upredicateBV{p}}{p \in \upredicates}$, obtained from $\Delta_{\BVshort}$ by replacing each uninterpreted predicate $p \in \upredicates$ with $\upredicateBV{p}$.

		\item $\IAfragment$, the \emph{integer fragment}, is the set of CHCs over the signature $\signatureIA \cup \upredicatesIA$, where $\upredicatesIA = \set{\upredicateIA{p}}{p \in \upredicates}$, obtained from $\Delta_{\IAshort}$ by transforming each CHC of the form $\forall \bar{x}.\;
        	\bigl( \phi \;\wedge\; \bigwedge_j p_j(\bar{x}_{p_j}) \bigr)
        	\;\rightarrow\;
            h(\bar{x}_h)$
        into
        \begin{equation*}
        	\forall \bar{x}'.\;
        	\bigl( \formulaBVIA(\phi)
        	\;\wedge\;
        	\bigwedge_j \upredicateIA{p}_j(\bar{x}'_{p_j})
            \;\wedge\;
        	\bigwedge_{x \in \bar{y}} \restrBV{x'} \bigr)
        	\;\rightarrow\;
        	\upredicateIA{h}(\bar{x}'_h)
        \end{equation*}
        where $\bar{y} = \bar{x} \setminus (\bigcup_j \bar{x}_{p_j})$.
        If the head is $\bot$, it remains unchanged.

		\item $\interfaces$, the set of \emph{interface constraints}, is
        $\interfacesBVIA \cup \interfacesIABV$, where
		\begin{itemize}
			\item $\interfacesBVIA$ consists of all pairs $\interface{\upredicateBV{p}}{\upredicateIA{p}}$ such that
            $\upredicateBV{p}$ is the head predicate of some CHC in $\BVfragment$ and $\upredicateIA{p}$ occurs in
            the body of some CHC in $\IAfragment$.
			\item $\interfacesIABV$ consists of all pairs $\interface{\upredicateIA{p}}{\upredicateBV{p}}$ such that
            $\upredicateIA{p}$ is the head predicate of some CHC in $\IAfragment$ and $\upredicateBV{p}$ occurs in
            the body of some CHC in $\BVfragment$.
		\end{itemize}
	\end{itemize}
\end{definition}

Note that the partition $(\Delta_{\BVshort},\Delta_{\IAshort})$ can be chosen arbitrarily, i.e., each CHC in $\chcs$ may be assigned to either the bit-vector or the integer fragment.
While different partitioning heuristics may affect performance, all resulting encodings preserve the semantics of $\chcs$ (see the following subsection).
The partitioning heuristic used in our implementation is discussed in Section~\ref{sec:implementation}.

\begin{example}
    Fig.~\ref{fig:opposite_sign}(b) shows a CHC set $\chcs$, and Fig.~\ref{fig:opposite_sign}(c) shows a theory-modular encoding of $\chcs$ in which the query is placed in the bit-vector fragment and all other CHCs are placed in the integer fragment.
    The interface constraint $\interface{\upredicateIA{q}}{\upredicateBV{q}}$ preserves the original dependencies in $\chcs$.
    %The only explicit domain restriction appears in the transformed fact, eliminating spurious derivations in the integer fragment.
\end{example}

\subsection{Semantics and Correctness}

For clarity, some of the forthcoming definitions reuse notation from the single-theory case (Section~\ref{subsec:chcs}).
Let $\encoding = (\BVfragment,\IAfragment,\interfaces)$ be a theory-modular encoding of $\chcs$.
The satisfiability of $\encoding$ is defined with respect to both $\theoryBV$ and $\theoryIA$, where each fragment is interpreted in its corresponding theory.
We denote $\signatureBVIA = (\signatureBV,\signatureIA)$, $\theoryBVIA = (\theoryBV,\theoryIA)$, and $\modelBVIA = (\modelBV,\modelIA)$.

We begin by defining the semantics of the uninterpreted predicates in $\encoding$.
As in the single-theory case, the bounded semantics of a predicate consists of all tuples derivable through a bounded derivation.
In the theory-modular setting, however, this semantics decomposes into two components: a \emph{local} component, denoted $\rightarrow$ (Eq.~\ref{eq:bounded_multi_semantics_local}), which unfolds the definition of the predicate using CHCs from its own fragment, and an \emph{external} component, denoted $\leadsto$ (Eq.~\ref{eq:bounded_multi_semantics_external}), which derives tuples from the corresponding predicate copy in the other fragment, as dictated by the interface constraints.
Importantly, the domain restrictions introduced by the encoding in the integer fragment (see Definition~\ref{def:theory_modular_encoding}) ensure that only tuples relevant to $\chcs$, namely those representing bit-vector tuples, can be derived in $\IAfragment$.
%We consider bounds of the form $(\locBoundBV, \locBoundIA, \extBound)$, where $\locBoundBV, \locBoundIA, \extBound \in \mathbb{N}$.
%A derivation under such a bound may take up to $\locBoundBV$ local steps using CHCs from $\BVfragment$, up to $\locBoundIA$ local steps using CHCs from $\IAfragment$, and up to $\extBound$ external steps using interface constraints from $\interfaces$.
%For two bounds $(\locBoundBV, \locBoundIA, \extBound)$ and $(\locBoundBV', \locBoundIA', \extBound')$, we write $(\locBoundBV, \locBoundIA, \extBound) \leq (\locBoundBV', \locBoundIA', \extBound')$ if $\locBoundBV \leq \locBoundBV'$, $\locBoundIA \leq \locBoundIA'$, and $\extBound \leq \extBound'$.
The bounded semantics is defined inductively on the bound.

\begin{definition}[Bounded Semantics]\label{def:bounded_multi_semantics}
    Let $\upredicatePH{p} \in \upredicatesBV \cup \upredicatesIA$ be an uninterpreted predicate with $\PH \in \{\BVshort,\IAshort\}$, and let $n \in \mathbb{N}$ be a bound.
    The \emph{bounded semantics} of $\upredicatePH{p}$ with respect to $\modelBVIA$ at bound $n$, denoted $\semantics{\upredicatePH{p}}{\modelBVIA}{n}$, is defined as $\semantics{\upredicatePH{p}}{\modelBVIA}{n} = \locsemantics{\upredicatePH{p}}{\modelBVIA}{n} \cup \extsemantics{\upredicatePH{p}}{\modelBVIA}{n}$ such that:
    {%\footnotesize
    \begin{align}
        %==========local==========
        \locsemantics{\upredicatePH{p}}{\modelBVIA}{n} &= \hspace*{-0.3cm}
        \bigcup_{(\forall\bar{x}. \; \body[\bar{x}] \rightarrow \upredicatePH{p}(\bar{x}_{\upredicatePH{p}})) \in \fragmentPH} 
        \semantics{\exists (\bar{x} \setminus \bar{x}_{\upredicatePH{p}}). \; \body[\bar{x}]}{\modelPH^*}{} \label{eq:bounded_multi_semantics_local} \\
        & \text{where } \modelPH^*=
        \begin{cases}
            \modelPH\{\upredicatePH{q}\mapsto \emptyset\}_{\upredicatePH{q}\in\upredicatesPH} & \textup{if } n = 0 \\ \vspace*{-0.25cm}\\
            \modelPH\{\upredicatePH{q}\mapsto \semantics{\upredicatePH{q}}{\modelBVIA}{n-1}\}_{\upredicatePH{q}\in\upredicatesPH} & \textup{if } n > 0
        \end{cases}
        \nonumber \\ \nonumber \\
        %==========external==========
        \extsemantics{\upredicatePH{p}}{\modelBVIA}{n} &= 
        \begin{cases}
            \emptyset & \textup{if } n = 0 \textup{ or } \interface{\upredicatePHprime{p}}{\upredicatePH{p}} \notin \interfacesPHprime \\ \vspace*{-0.25cm} \\
            \domPHprime{p} \big( \semantics{\upredicatePHprime{p}}{\modelBVIA}{n-1} \bigr) & \textup{if } n > 0 \textup{ and } \interface{\upredicatePHprime{p}}{\upredicatePH{p}} \in \interfacesPHprime
            %\domIABV{p} \big( \semantics{\upredicateIA{p}}{\modelBVIA}{n-1} \bigr) & \textup{if } \PH = \BVshort, n > 0, \interface{\upredicateIA{p}}{\upredicateBV{p}} \in \interfacesIABV \\ \vspace*{-0.25cm} \\
            %\domBVIA{p} \big(\semantics{\upredicateBV{p}}{\modelBVIA}{n-1} \bigr) & \textup{if } \PH = \IAshort, n > 0, \interface{\upredicateBV{p}}{\upredicateIA{p}} \in \interfacesBVIA  
        \end{cases} \label{eq:bounded_multi_semantics_external}
        \\ & \text{where } \PHprime \in \{\BVshort,\IAshort\} \setminus\{\PH\} \nonumber
    \end{align}
    }
\end{definition}

\begin{example}\label{example:bounded_semantics_multi_theory}
    Consider the theory-modular encoding in Fig.~\ref{fig:opposite_sign}, where all variables have sort $\sortBV{2}$.
    Then $(1,0,0) \in \locsemantics{\upredicateIA{p}}{\modelBVIA}{0}$ (first CHC in $\IAfragment$), $(1,1,-1) \in \locsemantics{\upredicateIA{p}}{\modelBVIA}{1}$ (second CHC in $\IAfragment$), $(1,-1) \in \locsemantics{\upredicateIA{q}}{\modelBVIA}{2}$ (third CHC in $\IAfragment$), and $(01,11) \in \extsemantics{\upredicateBV{q}}{\modelBVIA}{3}$ (interface constraint).
\end{example}

Recall the inverse domain transformer $\domIABV{p}$ is applied only to integer tuples in the image of $\domBVIA{p}$, namely those corresponding to bit-vector tuples.
The following lemma states that the application of $\domIABV{p}$ in Eq.~\ref{eq:bounded_multi_semantics_external} is well-defined.

\begin{lemma}\label{lemma:well-defined}
    Let $p \in \upredicates$ be an uninterpreted predicate that occurs in $\chcs$ and let $n \in \mathbb{N}$ be a bound.
    Then, $\semantics{\upredicateIA{p}}\modelBVIA{n} \subseteq\domBVIA{p}\!\bigl(\domain^{\modelBV}_{p}\bigr)$.
%    \begin{equation*}
%        \semantics{\upredicateIA{p}}{\modelBVIA}{\locBoundBV,\locBoundIA,\extBound} \subseteq \domBVIA{p}\!\bigl(\domain^{\modelBV}_{p}\bigr)
%    \end{equation*}
\end{lemma}

\begin{proof}
    By induction on the bound. 
    See Appendix~\ref{appendix:proofs}.
\end{proof}

\begin{comment}
    \begin{proof}
    We prove the claim by induction on the bound.
    For the local component, consider a CHC in $\IAfragment$ of the form $\forall \bar{x}'.\; \body(\bar{x}') \rightarrow \upredicateIA{p}(\bar{x}'_{\upredicateIA{p}})$.
    Let $\bar{a}$ be a tuple of integers derived for $\upredicateIA{p}$ at bound $(\locBoundBV,\locBoundIA,\extBound)$ using this rule, that is, a tuple $\bar{a}$ such that $\modelIA^*, \assignment{\bar{x}'_{\upredicateIA{p}}}{\bar{a}} \models \exists (\bar{x} \setminus \bar{x}_{\upredicatePH{p}}). \; \body[\bar{x}]$.
    Let $x_i \in \bar{x}_{\upredicatePH{p}}$.
    If $x_i$ occurs in an uninterpreted predicate application in $\body(\bar{x}')$ lies in the image of $\domBVIA{x_i}$ by the induction hypothesis.
    Otherwise, by Definition~\ref{def:theory_modular_encoding}, the syntactic domain restriction constraint for $x_i$, $\restrBV{x_i}$, was conjoined the body of the CHC during the encoding.
    Therefore, $\bar{a}$ lies in the image of $\domBVIA{p}$.
    For the external component, any tuple derived for $\upredicateIA{p}$ is obtained by applying $\domBVIA{p}$ to a tuple $\bar{v} \in \domain^{\modelBV}_{p}$ derived for $\upredicateBV{p}$, as dictated by the interface constraints.
    Hence, every externally derived tuple also belongs to $\domBVIA{p}(\domain^{\modelBV}_{p})$.
\end{proof}
\end{comment}

The bounded semantics are monotonic, i.e., for all $n \in \mathbb{N}$, $\semantics{\upredicatePH{p}}{\modelBVIA}{n} \subseteq \semantics{\upredicatePH{p}}{\modelBVIA}{n+1}$.
Taking the union over all bounds yields the (unbounded) semantics of $\upredicatePH{p}$.
The \emph{local semantics} of $\upredicatePH{p}$, denoted $\locsemantics{\upredicatePH{p}}{\modelBVIA}{}$, is $\bigcup_{n \in \mathbb{N}} \locsemantics{\upredicatePH{p}}{\modelBVIA}{n}$.
Similarly, the \emph{external semantics} of $\upredicatePH{p}$, denoted $\extsemantics{\upredicatePH{p}}{\modelBVIA}{}$, is $\bigcup_{n \in \mathbb{N}} \extsemantics{\upredicatePH{p}}{\modelBVIA}{n}$.
The \emph{semantics} of $\upredicatePH{p}$, denoted $\semantics{\upredicatePH{p}}{\modelBVIA}{}$, is $\locsemantics{\upredicatePH{p}}{\modelBVIA}{} \cup \extsemantics{\upredicatePH{p}}{\modelBVIA}{}$.

\begin{comment}
\begin{definition}[Unbounded Semantics]\label{def:unbounded_multi_semantics}
    Let $\upredicatePH{p} \in \upredicatesBV \cup \upredicatesIA$ be an uninterpreted predicate with $\PH \in \{\BVshort,\IAshort\}$.
    The \emph{unbounded semantics} of $\upredicatePH{p}$ with respect to $\modelBVIA$, denoted $\semantics{\upredicatePH{p}}{\modelBVIA}{}$, is defined as $\semantics{\upredicatePH{p}}{\modelBVIA}{} = \locsemantics{\upredicatePH{p}}{\modelBVIA}{} \cup \extsemantics{\upredicatePH{p}}{\modelBVIA}{}$ such that: 
    $\locsemantics{\upredicatePH{p}}{\modelBVIA}{} = \bigcup_{n \in \mathbb{N}}{\locsemantics{\upredicatePH{p}}{\modelBVIA}{n}}$ and
    $\extsemantics{\upredicatePH{p}}{\modelBVIA}{} = \bigcup_{n \in \mathbb{N}}{\extsemantics{\upredicatePH{p}}{\modelBVIA}{n}}$.
    %follows: \ $\semantics{\upredicatePH{p}}{\modelBVIA}{} = \bigcup_{n \in \mathbb{N}}{\semantics{\upredicatePH{p}}{\modelBVIA}{n}}$.
    
    %{\small
    %\begin{displaymath}
        %\semantics{\upredicatePH{p}}{\modelBVIA}{} = \bigcup_{(\locBoundBV,\locBoundIA,\extBound) \in \mathbb{N}^3}{\semantics{\upredicatePH{p}}{\modelBVIA}{\locBoundBV,\locBoundIA,\extBound}}
    %\end{displaymath}
    %}%
\end{definition}    
\end{comment}

The satisfiability of a theory-modular encoding is defined in terms of the semantics of its uninterpreted predicates.

\begin{definition}[Satisfiability]\label{def:multi_satisfiability}
    A theory-modular encoding $\encoding$ is \emph{satisfiable modulo $\theoryBVIA$} if, for each $\PH \in \{\BVshort,\IAshort\}$ and every query $\chc \in \fragmentPH$, it holds that $\modelPH\{\upredicatePH{p} \mapsto \semantics{\upredicatePH{p}}{\modelBVIA}{}\}_{\upredicatePH{p} \in \upredicatesPH} \models \chc$.
\end{definition}

%We now establish the following lemma and theorem, which together demonstrate the soundness of the multi-theory encoding.
A key property of a theory-modular encoding is that it preserves the semantics of the uninterpreted predicates of the original set of CHCs.
In particular, the semantics of each uninterpreted predicate $p$ occurring in $\chcs$ coincides with the union of the semantics of its copies in the encoding, where the semantics obtained in the integer fragment are transformed back to the bit-vector domain via the inverse domain transformer.
Intuitively, the encoding preserves derivability, up to the domain transformers, at the level of the tuples produced for each predicate.
This is formalized in Lemma~\ref{lemma3}, which in turn implies Theorem~\ref{theorem:semantic_preserving_encoding}.

\begin{lemma}\label{lemma3}
    Let $p \in \upredicates$ be an uninterpreted predicate that occurs in $\chcs$.
    Then,
    {
    %\small
    \begin{equation*}
        \semantics{p}{\modelBV}{}
        \;=\; 
        \semantics{\upredicateBV{p}}{\modelBVIA}{} 
        \;\cup\; 
        \domIABV{p}\!\bigl( \semantics{\upredicateIA{p}}{\modelBVIA}{} \bigr)
    \end{equation*}
    }%
\end{lemma}

%The semantics preservation of a theory-modular encoding yields the following theorem that demonstrates its soundness.

\begin{theorem}\label{theorem:semantic_preserving_encoding}
    $\chcs$ is satisfiable modulo $\theoryBV$ iff $\encoding$ is satisfiable modulo $\theoryBVIA$.
\end{theorem}

\begin{proof} See Appendix~\ref{appendix:proofs}. \end{proof}

As in the single-theory setting, our algorithm for checking the satisfiability of a theory-modular encoding (Section~\ref{sec:algorithm}) searches for symbolic interpretations of uninterpreted predicates that over-approximate their semantics while remaining strong enough to satisfy the queries.
A $\signatureBVIA$-\emph{interpretation} is a pair $\interpretationBVIA = (\interpretationBV, \interpretationIA)$, where $\interpretationBV$ is a $\signatureBV$-interpretation for $\BVfragment$ and $\interpretationIA$ is a $\signatureIA$-interpretation for $\IAfragment$.
Specifically, $\interpretationBV$ assigns to each uninterpreted predicate symbol $\upredicateBV{p} \in \upredicatesBV$ of type $\sortBV{n_1} \times \cdots \times \sortBV{n_m}$ a $\signatureBV$-formula whose free variables are in the tuple $\bar{y}_{\upredicateBV{p}} = (y_1,\dots,y_m)$, where each $y_i$ has sort $\sortBV{n_i}$.
Analogously, $\interpretationIA$ assigns to each $\upredicateIA{p} \in \upredicatesIA$ a $\signatureIA$-formula with free variables in the tuple $\bar{y}_{\upredicateIA{p}} = \bar{y}'_{\upredicateBV{p}}$, all of sort $\sortIA$.

\begin{definition}[Solution]\label{def:multi_solution}
    A $\signatureBVIA$-interpretation $\interpretationBVIA$ \emph{satisfies} $\encoding$ modulo $\theoryBVIA$, denoted $\interpretationBVIA \models_{\theoryBVIA} \encoding$, if the following conditions hold:
    \begin{itemize}
        \item \textbf{Local Satisfaction:}
        For each $\PH \in \{\BVshort,\IAshort\}$ and every CHC $\chc \in \fragmentPH$, the $\signaturePH$-formula $\interpretationPH(\chc)$ is valid modulo $\theoryPH$.
        
        \item \textbf{Interface Constraint Satisfaction:}
        \begin{itemize}
            \item For every interface constraint $\interface{\upredicateBV{p}}{\upredicateIA{p}} \in \interfacesBVIA$, the $\signatureBV$-formula
            \begin{equation*}
                \interpretationBVIA(\interface{\upredicateBV{p}}{\upredicateIA{p}})
                \;=\;
                \forall \bar{y}_{\upredicateBV{p}}\,.\;
                \interpretationBV(\upredicateBV{p})
                \rightarrow
                \formulaIABV \bigl(\interpretationIA(\upredicateIA{p})\bigr)
                \quad \text{is valid modulo $\theoryBV$.}
            \end{equation*}
            %is valid modulo $\theoryBV$.
            
            \item For every interface constraint $\interface{\upredicateIA{p}}{\upredicateBV{p}} \in \interfacesIABV$, the $\signatureBV$-formula
            \begin{equation*}
                \interpretationBVIA(\interface{\upredicateIA{p}}{\upredicateBV{p}})
                \;=\;
                \forall \bar{y}_{\upredicateBV{p}}\,.\;
                \formulaIABV \bigl(\interpretationIA(\upredicateIA{p})\bigr)
                \rightarrow
                \interpretationBV(\upredicateBV{p})
                \quad \text{is valid modulo $\theoryBV$.}
            \end{equation*}
            %is valid modulo $\theoryBV$.
        \end{itemize}
    \end{itemize}
    In this case, $\interpretationBVIA$ is called a \emph{$\signatureBVIA$-solution} for $\encoding$.
\end{definition}

Intuitively, \emph{local satisfaction} ensures that each fragment interpretation satisfies the CHCs of its own fragment in the corresponding theory, while \emph{interface constraint satisfaction} enforces consistency between the fragments.
For every interface constraint $\theta \in \interfaces$, the formula $\interpretationBVIA(\theta)$ requires that, within the restricted bit-vector domain, the interpretation of the target predicate includes all tuples corresponding to those admitted by the interpretation of the source predicate.
Notably, the interpretation of the integer copy of the predicate is allowed to admit additional tuples outside the restricted domain, and such tuples are discarded when $\formulaIABV$ is applied.

A $\signatureBVIA$-solution provides an over-approximation of the semantics of the uninterpreted predicates.
This property is established in Lemma~\ref{lemma:over_approximation}.

\begin{lemma}\label{lemma:over_approximation}
    Let $\interpretationBVIA$ be a $\signatureBVIA$-solution for $\encoding$, and let $\upredicatePH{p} \in \upredicatesBV \cup \upredicatesIA$ be an uninterpreted predicate with $\PH \in \{\BVshort,\IAshort\}$.
    Then, $\semantics{\upredicatePH{p}}{\modelBVIA}{}\subseteq \semantics{\interpretationPH(\upredicatePH{p})}{\modelPH}{}$.
    % {%\small
    % \begin{displaymath}
    %     \semantics{\upredicatePH{p}}{\modelBVIA}{}
    %     \;\subseteq\;
    %     \semantics{\interpretationPH(\upredicatePH{p})}{\modelPH}{}
    % \end{displaymath}
    % }
\end{lemma}

%The proof is given in Appendix~\ref{appendix:lemma_proof}.
From Lemma~\ref{lemma:over_approximation}, it follows that the existence of a $\signatureBVIA$-solution implies the satisfiability of the theory-modular encoding, as stated formally in Theorem~\ref{theorem:solution_implies_satisfiability}.

\begin{theorem}\label{theorem:solution_implies_satisfiability}
    %Let $\chcs$ be a CHC set over $\signatureBV \cup \upredicates$, and let $\encoding$ be a $\theories$-encoding of $\chcs$.
    If $\encoding$ admits a $\signatureBVIA$-solution, then $\encoding$ is satisfiable modulo $\theoryBVIA$.
\end{theorem}

The following theorem establishes a correspondence between solutions of the encoding and solutions of the original CHCs.
%The proof is given in Appendix~\ref{appendix:proofs}.

\begin{theorem}\label{theorem:solution_preservation}
    $\chcs$ admits a $\signatureBV$-solution if and only if $\encoding$ admits a $\signatureBVIA$-solution.
\end{theorem}

\begin{proof} See Appendix~\ref{appendix:proofs}. \end{proof}

\paragraph{Decidability.}
The encoding does not compromise the decidability of CHC-SAT modulo $\theoryBV$.
Intuitively, although the encoding introduces an integer component, all tuples that contribute to the semantics of the predicates correspond to values in the original finite bit-vector domain.
%In particular, by Lemma~\ref{lemma2}, the semantics derived for each integer predicate is contained in the image of the corresponding bit-vector domain under the domain transformers, and is therefore finite.
Thus, the encoding does not introduce additional sources of infiniteness that affect decidability.

\begin{comment}
The following corollary is an immediate consequence of Theorems~\ref{theorem:semantic_preserving_encoding} and~\ref{theorem:solution_implies_satisfiability}.

\begin{corollary}
    %Let $\chcs$ be a CHC set over $\signatureBV \cup \upredicates$, and let $\encoding$ be a $\theories$-encoding of $\chcs$.
    If $\encoding$ admits a $\signatureBVIA$-solution, then $\chcs$ is satisfiable modulo $\theoryBV$.
\end{corollary}
\end{comment}

%% file: sections/algorithm.tex
\section{Algorithm}\label{sec:algorithm}

This section presents $\solve$, a symbolic procedure for deciding the satisfiability of theory-modular encodings.
The procedure relies on two CHC satisfiability oracles: one for the bit-vector fragment modulo $\theoryBV$ and one for the integer fragment modulo $\theoryIA$.
At a high level, $\solve$ generates queries over a single theory and delegates them to the corresponding oracles, while coordinating the exchange of information between fragments.
This coordination is achieved by propagating constraints soundly across interface constraints using theory transformers, and by maintaining both under- and over-approximations of the external semantics of uninterpreted predicates.
The under-approximations are used to search for counterexamples to satisfiability, i.e., derivations leading to a query violation, while the over-approximations are used to block infeasible counterexamples and detect convergence.

The procedure is defined as a transition system with states specified below and transitions given by the inference rules in Fig.~\ref{fig:mosaic_solve}.
The input to $\solve$ is a theory-modular encoding $\encoding = (\BVfragment,\IAfragment,\interfaces)$.
For convenience, we assume that $\encoding$ contains no query and that a query $\chc_q$ is provided separately.
Without loss of generality, $\chc_q$ is a $\signatureBV$-formula of the form
$\forall \bar{y}_{\upredicateBV{q}}.\;
\bigl(\phi_q[\bar{y}_{\upredicateBV{q}}] \wedge \upredicateBV{q}(\bar{y}_{\upredicateBV{q}})\bigr) \rightarrow \bot$.
The goal of $\solve$ is to decide whether $(\BVfragment \cup \{\chc_q\}, \IAfragment, \interfaces)$ is satisfiable.

A state of $\solve$ is a triple $\state{\pobs}{\ulemmas}{\olemmas}$, where $\pobs$ is a finite set of \emph{proof obligations}, $\ulemmas$ is a set of under-approximation lemmas, and $\olemmas$ is a set of over-approximation lemmas.
Each proof obligation has the form $\pob{\upredicatePH{p}}{\phi}$, where $\PH \in \{\BVshort,\IAshort\}$, $\upredicatePH{p} \in \upredicatesPH$ is an uninterpreted predicate in the fragment $\fragmentPH$, and $\phi$ is a quantifier-free $\signaturePH$-formula.
A proof obligation $\pob{\upredicatePH{p}}{\phi}$ is generated when $\solve$ aims to determine whether the semantics of $\upredicatePH{p}$ contains a tuple satisfying $\phi$.
Intuitively, $\phi$ characterizes a set of \emph{bad} tuples, whose presence in the semantics of $\upredicatePH{p}$ may witness a counterexample to satisfiability.
The sets $\ulemmas$ and $\olemmas$ associate, with each predicate $\upredicatePH{p}$ that is the target of an interface constraint $\interface{\upredicatePHprime{p}}{\upredicatePH{p}} \in \interfaces$ (where $\PH \neq \PHprime$), a set of learned lemmas.
Intuitively, if $\icformula{\upredicatePH{p}}{\psi}$ is stored in $\ulemmas$ (resp.\ $\olemmas$), then $\psi$ under- (resp.\ over-) approximates $\extsemantics{\upredicatePH{p}}{\modelBVIA}{}$, i.e., the set of tuples derivable for $\upredicatePH{p}$ via its corresponding predicate copy $\upredicatePHprime{p}$ in the other fragment (see Definition~\ref{def:bounded_multi_semantics}).

\input{algorithms/bounded_sat_chc_solvers}

The procedure $\solve$ is initialized with a single proof obligation corresponding to the input query $\chc_q$, and with empty lemma sets (rule \textsc{Init}).
The rules \textsc{NegAnswer}, \textsc{PosAnswer}, and \textsc{Query} are then applied repeatedly.
Rules \textsc{NegAnswer} and \textsc{PosAnswer} are applied when a proof obligation in $\pobs$ can be resolved using the current over- and under-approximation lemmas in $\olemmas$ and $\ulemmas$, respectively.
Otherwise, if the current approximations are insufficient, rule \textsc{Query} is applied to generate a new proof obligation in the other fragment.
When the initial proof obligation is resolved using $\olemmas$ or $\ulemmas$, $\solve$ terminates with SAT or UNSAT, respectively (rules \textsc{Sat} and \textsc{Unsat}).
We explain the rules \textsc{NegAnswer}, \textsc{PosAnswer}, and \textsc{Query} below.

\medskip
\noindent\textbf{\textsc{NegAnswer}} discharges a proof obligation $\pob{\upredicatePH{p}}{\phi} \in \pobs$, other than the initial one, if it can be answered negatively using the current over-approximation.
To establish this, the CHC satisfiability oracle for $\theoryPH$ is invoked $\fragmentPHover \cup \Query(\upredicatePH{p},\phi)$ (defined at the bottom of Fig.~\ref{fig:mosaic_solve}).
Here, $\fragmentPHover$ extends $\fragmentPH$ with facts that encode, for each predicate that is the target of an interface constraint, the current over-approximation of its external semantics by conjoining the corresponding lemmas in $\olemmas$.
The query $\Query(\upredicatePH{p},\phi)$ asserts that the semantics of $\upredicatePH{p}$ is disjoint from the set of tuples satisfying $\phi$.
If $\PH = \IAshort$, syntactic domain restrictions are applied to ensure that all derivations correspond to the bit-vector domain.
If the above CHC set is satisfiable, the oracle returns a $\signaturePH$-solution $\interpretationPH$ such that $\interpretationPH(\upredicatePH{p})$ over-approximates $\semantics{\upredicatePH{p}}{\modelBVIA}{}$ and $\models_{\theoryPH} \interpretationPH(\upredicatePH{p}) \rightarrow \neg \phi$.
Consequently, the proof obligation is discharged and removed from $\pobs$.
Furthermore, the over-approximation of the corresponding predicate $\upredicatePHprime{p}$ in the other fragment is strengthened with a new lemma obtained by translating $\interpretationPH(\upredicatePH{p})$ via the appropriate theory transformer.
%Since $\interpretationPH(\upredicatePH{p})$ over-approximates the semantics of $\upredicatePH{p}$, the learned lemma over-approximates the external semantics of $\upredicatePHprime{p}$.

\medskip
\noindent\textbf{\textsc{PosAnswer}} discharges a proof obligation $\pob{\upredicatePH{p}}{\phi} \in \pobs$, other than the initial one, if it can be answered positively using the current under-approximation.
To establish this, the CHC satisfiability oracle for $\theoryPH$ is invoked on $\fragmentPHunder \cup \Query(\upredicatePH{p},\phi)$, which is constructed as in rule \textsc{NegAnswer}, except that the lemmas in $\ulemmas$ are used instead of those in $\olemmas$ and are combined disjunctively.
If the above CHC set is unsatisfiable, the oracle returns a refutation $\refutation$, i.e., a derivation of $\bot$.
The child of the root $\bot$ in $\refutation$ is a pair $\langle \upredicatePH{p}, \psi \rangle$, where $\psi$ is a $\signaturePH$-formula that under-approximates $\semantics{\upredicatePH{p}}{\modelBVIA}{}$ and satisfies $\not\models_{\theoryPH} \psi \rightarrow \neg \phi$.
Consequently, the proof obligation is discharged and removed from $\pobs$.
Furthermore, the under-approximation of the corresponding predicate $\upredicatePHprime{p}$ in the other fragment is extended with a new lemma obtained by translating $\psi$ via the appropriate theory transformer.

\medskip
\noindent\textbf{\textsc{Query}} generates a new proof obligation when an existing proof obligation $\pob{\upredicatePH{p}}{\phi} \in \pobs$ cannot be resolved using the current approximations.
In this case, $\fragmentPHunder \cup \Query(\upredicatePH{p},\phi)$ is satisfiable, whereas $\fragmentPHover \cup \Query(\upredicatePH{p},\phi)$ is unsatisfiable.
Thus, the current under-approximation does not witness a counterexample to the proof obligation, while the current over-approximation does.
The CHC satisfiability oracle for $\theoryPH$ returns a refutation $\refutation$ for $\fragmentPHover \cup \Query(\upredicatePH{p},\phi)$.
Since $\refutation$ is obtained using the over-approximation in $\olemmas$, it is not yet clear whether it corresponds to a feasible derivation of the original encoding $\encoding$.
Therefore, a leaf $\langle \upredicatePH{r}, \psi \rangle$ of $\refutation$ derived from a fact contributed by $\olemmas$ %(i.e., not from $\fragmentPH$)
is selected.
Such a leaf must exist, since otherwise all leaves would be derived from facts in $\fragmentPH$, implying that $\fragmentPHunder \cup \Query(\upredicatePH{p},\phi)$ is unsatisfiable.
To determine whether there exists a tuple in $\extsemantics{\upredicatePH{r}}{\modelBVIA}{}$ that satisfies $\psi$, a proof obligation is generated for the corresponding predicate $\upredicatePHprime{r}$ in the other fragment by translating $\psi$ via the appropriate theory transformer.

%\medskip
%\noindent We conclude with a soundness result and a remark on termination.

\begin{theorem}[Soundness]\label{theorem:soundness}
    Let $\encoding=(\BVfragment,\IAfragment,\interfaces)$ be a theory-modular encoding and $\chc_q$ a $\signatureBV$-query.
    If $\solve$ returns \emph{SAT} (resp.\ \emph{UNSAT}), then $(\BVfragment \cup \{\chc_q\}, \IAfragment, \interfaces)$ is satisfiable (resp.\ unsatisfiable).
\end{theorem}

\begin{proof}
    By induction on the length of the transition sequence.
    See Appendix~\ref{appendix:proofs}.
\end{proof}

\paragraph{Termination.}
The presentation of $\solve$ omits several refinements that ensure termination.
First, when applying rule \textsc{Query}, let $\langle \upredicatePH{r}, \psi \rangle$ be the selected leaf, which induces a new proof obligation.
The following two conditions must hold:
(i) the new proof obligation does not overlap with any existing proof obligation, and
(ii) no tuple in the current under-approximation of $\upredicatePH{r}$ satisfies $\psi$.
These conditions can be checked by examining the existing proof obligations in $\pobs$ corresponding to $\upredicatePHprime{r}$ and the under-approximation lemmas in $\ulemmas$ corresponding to $\upredicatePH{r}$.
Second, the procedure follows an iterative deepening strategy, similar to that of~\cite{DBLP:journals/fmsd/KomuravelliGC16}, which bounds the exploration of derivations and ensures that all relevant obligations are eventually considered.
Under these refinements, and assuming each invocation of the underlying CHC satisfiability oracles terminates, $\solve$ is guaranteed to terminate.

%% file: algorithms/bounded_sat_chc_solvers.tex
\begin{figure}[!t]
    \centering
    \scalebox{0.93}{
    \begin{mathpar}
        %================= Init ================
        \inferrule*[left=Init]
        {\\} 
        {
            \state
                {\{\, \pob{\upredicateBV{q}}{\phi_q} \,\}}
                {\emptyset}
                {\emptyset}
        }
        \vspace{0.275em}
        
        %============== NegAnswer ==============
        \inferrule*[left=NegAnswer]
        {
            \state{\pobs}{\ulemmas}{\olemmas} \\
            \pob{\upredicatePH{p}}{\phi} \in \pobs \setminus \{\, \pob{\upredicateBV{q}}{\phi_q} \,\} \\
            \PH \neq \PHprime \in \{ \BVshort,\IAshort \} \\\\
            \interpretationPH \models_{\theoryPH} \bigl( \fragmentPHover \cup \Query(\upredicatePH{p},\phi) \bigr) 
        } 
        {
            \state
                {\pobs \setminus \{ \pob{\upredicatePH{p}}{\phi} \}} 
                {\ulemmas}
                {\olemmas \cup \{\,\icformula{\upredicatePHprime{p}}{\formulaPH(\interpretationPH(\upredicatePH{p}))} \,\}}
        }
        \vspace{0.275em}

        %============== PosAnswer ==============
        \inferrule*[left=PosAnswer]
        {
            \state{\pobs}{\ulemmas}{\olemmas} \\
            \pob{\upredicatePH{p}}{\phi} \in \pobs \setminus \{\, \pob{\upredicateBV{q}}{\phi_q} \,\} \\
            \PH \neq \PHprime \in \{ \BVshort,\IAshort \} \\\\
            \bigl( \fragmentPHunder \cup \Query(\upredicatePH{p},\phi)\bigr) \vdash_{\theoryPH} \bot \quad \text{with refutation } \refutation \\
            \langle \, \upredicatePH{p},\psi \, \rangle = \refutationchild{\refutation}
        } 
        {
            \state
                {\pobs \setminus \{ \pob{\upredicatePH{p}}{\phi} \}} 
                {\ulemmas \cup \{\,\icformula{\upredicatePHprime{p}}{\formulaPH(\psi)} \,\}}
                {\olemmas}
        }
        \vspace{0.275em}
        
        %================ Query ================
        \inferrule*[left=Query]
        {
            \state{\pobs}{\ulemmas}{\olemmas} \\
            \pob{\upredicatePH{p}}{\phi} \in \pobs \\
            \PH \neq \PHprime \in \{ \BVshort,\IAshort \} \\\\
            \interpretationPH \models_{\theoryPH} \bigl( \fragmentPHunder \cup \Query(\upredicatePH{p},\phi) \bigr) \\
            \bigl( \fragmentPHover \cup \Query(\upredicatePH{p},\phi) \bigr) \vdash_{\theoryPH} \bot \quad \text{with refutation } \refutation \\
            \langle \, \upredicatePH{r},\psi \, \rangle \in \refutationleaves{\refutation} \\
            \interface{\upredicatePHprime{r}}{\upredicatePH{r}} \in \interfacesPHprime
            %\models_{\theoryPH} \psi \rightarrow \neg \bigl( \hspace*{-0.5cm} \bigvee_{\icformula{\upredicatePH{r}}{\psi'} \in \ulemmas} \hspace*{-0.5cm} \psi' \bigr) \\
            %\models_{\theoryPHprime} \formulaPH(\psi) \rightarrow \neg \bigl( \hspace*{-0.5cm} \bigvee_{\pob{\PHprime}{\upredicatePHprime{r}}{\psi'} \in \pobs} \hspace*{-0.5cm} \psi' \bigr)
        } 
        {
            \state
                {\pobs \cup \{ \pob{\upredicatePHprime{r}}{\formulaPH(\psi)} \}} 
                {\ulemmas}
                {\olemmas}
        }
        \vspace{0.275em}

        %================= SAT =================
        \inferrule* [left=Sat] 
        {
            \state{\{\, \pob{\upredicateBV{q}}{\phi_q} \,\}}{\ulemmas}{\olemmas} \\\\
            \interpretationBV \models_{\theoryBV} \bigl( \BVfragmentover \cup \Query(\upredicateBV{q},\phi_q) \bigr)
        } 
        {
            {\text{SAT}}
        }
        \hspace{1.5em}
        %================ UNSAT ================
        \inferrule* [left=Unsat] 
        {
            \state{\{\, \pob{\upredicateBV{q}}{\phi_q} \,\}}{\ulemmas}{\olemmas} \\\\
            \bigl( \BVfragmentunder \cup \Query(\upredicateBV{q},\phi_q) \bigr) \vdash_{\theoryBV} \bot
        } 
        {
            {\text{UNSAT}}
        }
    \end{mathpar}
    }
    \scalebox{0.95}{
        \begin{minipage}{\linewidth}
            \begin{align*}
                & \fragmentPHover \;=\; \fragmentPH \cup \bigcup_
                    {\interface{\upredicatePHprime{s}}{\upredicatePH{s}} \in \interfacesPHprime}
                    {\bigl\{ \,
                        \forall \bar{y}_{\upredicatePH{s}}.\; \bigl( \hspace*{-0.35cm} \bigwedge_{\icformula{\upredicatePH{s}}{\psi} \in \olemmas} \hspace*{-0.35cm} \psi[\bar{y}_{\upredicatePH{s}}] \;\wedge\; \restrictvars{\PH}{\bar{y}_{\upredicatePH{s}}} \bigr) \rightarrow \upredicatePH{s}(\bar{y}_{\upredicatePH{s}})
                    \, \bigr\}} \\
                & \fragmentPHunder \;=\; \fragmentPH \cup \bigcup_
                    {\interface{\upredicatePHprime{s}}{\upredicatePH{s}} \in \interfacesPHprime}
                    {\bigl\{ \,
                        \forall \bar{y}_{\upredicatePH{s}}.\; \bigl( \hspace*{-0.35cm} \bigvee_{\icformula{\upredicatePH{s}}{\psi} \in \ulemmas} \hspace*{-0.35cm} \psi[\bar{y}_{\upredicatePH{s}}] \;\wedge\; \restrictvars{\PH}{\bar{y}_{\upredicatePH{s}}} \bigr) \rightarrow \upredicatePH{s}(\bar{y}_{\upredicatePH{s}})
                    \, \bigr\}} \\
                & \Query(\upredicatePH{p},\phi) = \bigl\{ \,
                    \forall\bar{y}_{\upredicatePH{p}}. \; \bigl( \phi[\bar{y}_{\upredicatePH{p}}] \wedge \upredicatePH{p}(\bar{y}_{\upredicatePH{p}}) \wedge \restrictvars{\PH}{\bar{y}_{\upredicatePH{p}}} \bigr) \rightarrow \bot
                    \, \bigr\} \\
                %& \restrictvars{\IAshort}{\bar{y}} = \bigwedge_{y \in \bar{y}} \restrBV{y} 
                    %\qquad
                    %\restrictvars{\BVshort}{\bar{y}} = \top
                & \restrictvars{\PH}{\bar{y}} =
                \begin{cases}
                    \bigwedge_{y' \in \bar{y}} \restrBV{y'} & \text{if } \PH = \IAshort \\
                    \top & \text{if } \PH = \BVshort
                \end{cases}
            \end{align*}
        \end{minipage}
    }
    \caption{Inference rules defining $\solve$.}
    \label{fig:mosaic_solve}
    \vspace{-10pt}
\end{figure}

%% file: sections/implementation.tex
\section{Implementation Details and Experimental Results}\label{sec:implementation}

This section describes the implementation of the \mths framework and its evaluation.
We present the theory transformers used in our implementation, the strategy for partitioning a CHC set into theory fragments, and an optimization applied during the encoding.
Our prototype is implemented using Z3~\cite{DBLP:conf/tacas/MouraB08} and \spacer~\cite{DBLP:journals/fmsd/KomuravelliGC16}, where \spacer serves as the CHC-SAT oracle for both $\theoryBV$ and $\theoryIA$.

%\subsection{Implementation Details}\label{subsec:implementation_theory_transformers}

\paragraph{\bf Theory Transformers.} 
%\yv{Can't we just say something like: ``There are two natural choices when mapping elements in the domain of a variable $x$ of sort $\sortBV{n}$ to elements in the domain of its corresponding variable $x'$ of sort $\sortIA$ (see Definition~\ref{def:domain_transformers}): either...'' -- it will save a lot of space to just say what you want to say :) Moreover, since you already covered signed and unsigned transformers earlier, you can cut this description further here.}
Recall that translations are defined with respect to domain transformers (Definition~\ref{def:domain_transformers}), which map elements in the domain of each bit-vector variable to elements in the domain of its corresponding integer variable.
Two natural choices are signed and unsigned semantics.
In programming languages, signedness is typically part of a variable’s type (e.g., \texttt{int} vs.\ \texttt{unsigned int} in C).
In contrast, in $\signatureBV$, the sort of a variable depends only on its bit-width, while signedness is reflected in certain predicates and functions (e.g., $\bvult$ vs.\ $\bvslt$).
Consequently, the same variable may appear in both signed and unsigned contexts.
We adopt the following strategy. 
For each bit-vector variable, if it occurs more frequently in unsigned than signed predicate or function applications in $\chcs$, we use the unsigned domain transformer. 
Otherwise, we use the signed one.

We now describe the formula transformers $\formulaBVIA$ and $\formulaIABV$ (Definition~\ref{def:theory_transformers}).
For $\formulaBVIA$, we use the translation of Zohar et al.~\cite{DBLP:conf/vmcai/ZoharIMNNPRBT22}, which maps $\qfBV$ formulas to $\qfIA$ formulas and satisfies Eq.~\ref{eq:theory_transformer_BV_IA}.
Since the translation in~\cite{DBLP:conf/vmcai/ZoharIMNNPRBT22} assumes unsigned semantics, we extend it to signed semantics by adjusting domain restrictions and inserting the appropriate signed-unsigned conversions.

For $\formulaIABV$, we define a custom transformer based on syntactic substitution.
Arithmetic operators are mapped to their bit-vector counterparts (e.g., $+ \mapsto \bvadd$), predicates to signed bit-vector predicates (e.g., $< \mapsto \bvslt$), variables according to the variable mapping, and each integer constant $c$ to a bit-vector constant whose signed value is $c$.
However, this substitution alone does not satisfy Eq.~\ref{eq:theory_transformer_IA_BV}.
For example, consider the $\signatureIA$-formula $\varphi' = x' < y' + 1$, where the corresponding bit-vector variables $x,y$ have sort $\sortBV{2}$ and use the signed domain transformer.
The syntactic translation yields $\varphi = x \bvslt (y \bvadd 01)$.
The assignment $x' \mapsto 1, y' \mapsto 1$ satisfies $\varphi'$, but the corresponding assignment $x \mapsto 01, y \mapsto 01$ does not satisfy $\varphi$, since $01 \bvadd 01 = 10$, whose signed value is $-2$.
This mismatch is caused by overflow: bit-vector arithmetic is modular and operations may wrap around.
To avoid this, and to obtain a $\signatureBV$-formula that agrees with the original $\signatureIA$-formula on the restricted bit-vector domain, we extend variables and constants with extra bits using sign extension, choosing the minimal width for which no arithmetic term can overflow.
In this example, one extra bit suffices.

%\subsection{Encoding}\label{subsec:implementation_encoding}

\paragraph{\bf Encoding.} 
Given a CHC set $\chcs$ over $\signatureBV$, the partitioning strategy \mths applies (Definition~\ref{def:theory_modular_encoding}) is simple: CHCs whose function applications consist only of arithmetic operations ($\bvadd$, $\bvsub$, $\bvmul$, $\bvudiv$, $\bvurem$) are placed in the integer fragment, while all other CHCs remain in the bit-vector fragment.
The intuition is that arithmetic constraints are often easier to reason about in $\theoryIA$ than in $\theoryBV$, whereas bit-precise operations are better handled in $\theoryBV$, since translating them to $\theoryIA$ produces complex terms involving many $\arithdiv$ and $\arithmod$ operations.

We apply an additional encoding optimization.
The transformer $\formulaBVIA$ preserves bit-vector modular semantics by introducing modular arithmetic into the translated integer formulas (e.g., $x \bvadd y$ is translated to $(x' + y') \bmod 2^{n}$, where $n$ is the width of $x$ and $y$).
However, for CHCs placed in $\IAfragment$, if a term is guaranteed not to overflow, the modulo operation can be safely omitted.
For instance, in the CHCs of Fig.~\ref{fig:opposite_sign}(b), which encode the program in Fig.~\ref{fig:opposite_sign}(a), overflow is impossible along every execution.\footnote{This holds for all benchmarks used in our evaluation.}
Accordingly, the terms $a \bvadd 1$ and $b \bvsub 1$ are translated simply as $a' + 1$ and $b' - 1$, without the modulo operation.
%This optimization yields integer encodings that are better suited for reasoning in $\theoryIA$.

%% file: sections/evaluation.tex
%\subsection{Experimental Evaluation}\label{sec:evaluation}
\paragraph{\bf Experimental Evaluation.} 

%We implemented a prototype of \mths using Z3~\cite{DBLP:conf/tacas/MouraB08} and \spacer~\cite{DBLP:journals/fmsd/KomuravelliGC16}. our implementation includes the theory transformers for $\theoryBV$ and $\theoryIA$ as described above. \tmsat uses \spacer as the CHC-SAT oracle for both $\theoryBV$ and $\theoryIA$. %One limitation of our prototype is the fact that the multi-theory encoding is done manually.

%For the evaluation, we collected 11 small realistic programs (see Appendix~\ref{appendix:benchmarks}) and reduced them to CHCs modulo \qfbv. These programs perform bit-wise operations to implement functions like \texttt{max} and \texttt{abs} to achieve efficiency and security\footnote{\url{https://graphics.stanford.edu/~seander/bithacks.html}}.
%We created a benchmark set, which is derived from small realistic programs that perform bit-wise operations. There exist a set of well known functions (e.g. \texttt{max}, \texttt{abs}) that are implemented using bit-wise operations in order to achieve efficiency (and at times security guarantees)\footnote{\url{https://graphics.stanford.edu/~seander/bithacks.html}}. We use these known operations and encode programs that use them using CHCs modulo \qfbv. We also manually perform the multi-theory encoding for these examples.
%Our implementation, as well as the benchmarks used for the evaluation are available in an open-source repository.\footnote{Omitted for blind-review.}%\url{https://github.com/TechnionFV/MultiTheoryHorn}}.

\begin{wraptable}{r}{0.51\textwidth}
\vspace{-20pt}
\centering
\scalebox{0.7}{
\renewcommand{\arraystretch}{1}
\setlength{\tabcolsep}{6pt}
\sisetup{
  table-number-alignment = center,
  round-mode = places,
  round-precision = 2
}
\begin{tabular}{
l
c S[table-format=4.2]
c S[table-format=4.2]
c S[table-format=4.2]
}
\toprule
& \multicolumn{2}{c}{BV}
& \multicolumn{2}{c}{IA}
& \multicolumn{2}{c}{BV+IA} \\
\cmidrule(lr){2-3} \cmidrule(lr){4-5} \cmidrule(lr){6-7}
& \multicolumn{1}{c}{W} & \multicolumn{1}{c}{Time [s]}
& \multicolumn{1}{c}{W} & \multicolumn{1}{c}{Time [s]}
& \multicolumn{1}{c}{W} & \multicolumn{1}{c}{Time [s]} \\
\midrule
\textsf{abs-ge}         & 7 & 693.8             & 31 & 0.21         & \textbf{63} & 0.18    \\
\textsf{abs-sum}        & \textbf{6} & 631.82   & 3  & 37.23        & \textbf{6}  & 246.41  \\
\textsf{cond-neg}       & 7 & 76.17             & 8  & 506.52       & \textbf{63} & 0.47    \\
\textsf{cond-neg-diff}  & 6 & 1308.05           & 4  & 2723.5       & \textbf{8}  & 639.8   \\
\textsf{max-inv}        & 7 & 5122.19           & 4  & 0.05         & \textbf{8}  & 6188.19 \\
\textsf{opp-signs}      & 9 & 4358.99           & 3  & 13.01        & \textbf{62} & 0.09    \\
\textsf{opp-signs-diff} & \textbf{6} & 533.03   & 3  & 330.25       & \textbf{6}  & 548.96  \\
\textsf{swap}           & 7 & 2031.24           & 3  & 55.68        & \textbf{11} & 2142.14 \\
\textsf{swap-sum}       & \textbf{7} & 2109.57  & 4  & 243.96       & 6  & 134.21  \\
\textsf{turn-off-rm}    & \textbf{7} & 1718.3   & 4  & 235.95       & \textbf{7}  & 111.17  \\
\textsf{turn-on-lsb}    & \textbf{6} & 247.1    & 4  & 774.95       & \textbf{6}  & 115.98  \\
\bottomrule
\end{tabular}
}
\vspace{-5pt}
\caption{Performance of BV, IA, and BV+IA. Maximal bit-width (W) and runtime. Bold indicates the largest W.}\label{tbl:results}
\vspace{-15pt}
\end{wraptable}

%We use our generated benchmark to evaluate \mths. 

For the evaluation, we collected 11 small realistic programs (see Appendix~\ref{appendix:benchmarks}) and encoded them as CHCs modulo $\theoryBV$.
These programs use bit-wise operations to implement functions such as \texttt{max} and \texttt{abs} for efficiency and security\footnote{\url{https://graphics.stanford.edu/~seander/bithacks.html}}.
Experiments were conducted on a workstation with AMD EPYC 74F3. 
Every test was given 7200 seconds and 16GB of memory.
Our implementation and the benchmarks are publicly available.\footnote{\url{https://doi.org/10.5281/zenodo.21790677}}

We evaluate three configurations of \mths that differ in the partitioning strategy:
(i) BV, where all CHCs are placed in the bit-vector fragment (equivalent to running \spacer modulo $\theoryBV$);
(ii) IA, where all CHCs are translated to $\theoryIA$ and placed in the integer fragment (equivalent to running \spacer modulo $\theoryIA$ on the translated set); and
(iii) BV+IA, our theory-modular approach.
For each configuration, we report the maximal bit-width solved within the timeout and the runtime at that width (Table~\ref{tbl:results}).
Overall, BV+IA achieves the largest bit-width on most benchmarks, often by a wide margin.
In several cases, it scales to large bit-widths (up to 63), while BV and IA time out at much smaller widths.
When the maximal bit-width is the same, BV+IA typically improves runtime, sometimes by orders of magnitude.
The IA configuration is generally less effective, despite the arithmetic structure of the benchmarks, due to the translation of bit-wise operations into complex arithmetic terms involving many $\arithdiv$ and $\arithmod$ operations.
In contrast, BV often fails to scale to larger bit-widths.
Note that Z3 reports an internal error beyond bit-width 64, limiting the evaluation range.

%% file: sections/conclusion.tex
\section{Conclusion}\label{sec:conclusion}

We presented \mths, a theory-modular framework for deciding satisfiability of CHCs modulo $\theoryBV$.
It reduces CHCs into two fragments interpreted over $\theoryBV$ and $\theoryIA$, connected by interface constraints.
We introduced a procedure for determining satisfiability of such encodings and implemented a prototype.
Our evaluation shows that the approach can significantly outperform state-of-the-art solvers.

%% file: appendix/proofs.tex
\section{Additional Proofs}\label{appendix:proofs}

\paragraph{\bf Proof of Lemma~\ref{lemma:well-defined}.}
The claim is proved by induction on the bound.
For the local component, consider a CHC in $\IAfragment$ with head $\upredicateIA{p}$.
Any tuple derived using this CHC assigns to each head variable either (i) a value obtained from an uninterpreted predicate application in the body, which by the induction hypothesis lies in the image of the corresponding domain transformer, or (ii) a value satisfying the syntactic domain restriction $\restrBV{\cdot}$ conjoined during the encoding (Definition~\ref{def:theory_modular_encoding}).
In both cases, the derived tuple lies in the image of $\domBVIA{p}$.
For the external component, tuples are obtained by applying $\domBVIA{p}$ to tuples derived for $\upredicateBV{p}$, as dictated by the interface constraints.
Hence, all externally derived tuples lie in $\domBVIA{p}(\domain^{\modelBV}_{p})$. \qed

\paragraph{\bf Proof of Lemma~\ref{lemma3}.}
Throughout the proof, let $\chc = \forall \bar{x}.\, \bigl( \varphi \wedge \bigwedge_j p_j(\bar{x}_{p_j}) \bigr) \rightarrow H[\bar{x}]$ range over CHCs in $\chcs$, with body $\body[\bar{x}]$, and let $\upredicateBV{\chc}$ and $\upredicateIA{\chc}$ denote its copies per Definition~\ref{def:theory_modular_encoding} (when $\chc \in \Delta_{\BVshort}$ and $\chc \in \Delta_{\IAshort}$, respectively), with bodies $\body^{\BVshort}$ and $\body^{\IAshort}$.
Since the body of a CHC is a conjunction of a constraint and uninterpreted predicate applications, $\semantics{\exists (\bar{x} \setminus \bar{x}_h).\, \body[\bar{x}]}{\model\{p_j \mapsto R_j\}}{}$ is monotone in the relations $R_j$, for any tuple $\bar{x}_h \subseteq \bar{x}$.
We use the following two claims.
\begin{itemize}
    \item[(A)] Let $R_j \subseteq \domain^{\modelBV}_{p_j}$ and $R'_j = \domBVIA{p_j}(R_j)$ for each $j$, let $\chc \in \Delta_{\IAshort}$, and let $\bar{x}_h \subseteq \bar{x}$.
    Then
    \begin{displaymath}
        \domIABV{\bar{x}_h}\Bigl( \semantics{\exists (\bar{x}' \setminus \bar{x}'_h).\, \body^{\IAshort}[\bar{x}']}{\modelIA\{\upredicateIA{p}_j \mapsto R'_j\}}{} \Bigr)
        \;=\;
        \semantics{\exists (\bar{x} \setminus \bar{x}_h).\, \body[\bar{x}]}{\modelBV\{p_j \mapsto R_j\}}{}
    \end{displaymath}
    Indeed, every assignment satisfying $\body^{\IAshort}$ lies in the image of the domain transformers, since the relations $R'_j$ constrain the variables $\bigcup_j \bar{x}'_{p_j}$ to the image and the syntactic domain restrictions conjoined by the encoding constrain the remaining variables.
    Since domain transformers are defined per variable (with the consistency requirement guaranteeing that variables in the same argument position agree), such an assignment satisfies $\formulaBVIA(\varphi)$ iff its preimage satisfies $\varphi$ (Eq.~\ref{eq:theory_transformer_BV_IA}), and satisfies $\upredicateIA{p}_j(\bar{x}'_{p_j})$ under $R'_j$ iff its preimage satisfies $p_j(\bar{x}_{p_j})$ under $R_j$.
    Finally, projection onto $\bar{x}'_h$ commutes with the domain transformers.
    For $\chc \in \Delta_{\BVshort}$, the analogous equality
    $\semantics{\exists (\bar{x} \setminus \bar{x}_h).\, \body^{\BVshort}[\bar{x}]}{\modelBV\{\upredicateBV{p}_j \mapsto R_j\}}{} = \semantics{\exists (\bar{x} \setminus \bar{x}_h).\, \body[\bar{x}]}{\modelBV\{p_j \mapsto R_j\}}{}$
    is immediate, as $\upredicateBV{\chc}$ differs from $\chc$ only in the renaming of its uninterpreted predicates.
    \item[(B)] Let $\PH \in \{\BVshort,\IAshort\}$ and $\PHprime \in \{\BVshort,\IAshort\} \setminus \{\PH\}$.
    If $\upredicatePH{p}$ occurs in the body of some CHC in $\fragmentPH$, then $\domPHprime{p}\bigl(\semantics{\upredicatePHprime{p}}{\modelBVIA}{}\bigr) \subseteq \semantics{\upredicatePH{p}}{\modelBVIA}{}$.
    Indeed, if $\upredicatePHprime{p}$ is the head predicate of some CHC in $\chcs_{\PHprime}$, then $\interface{\upredicatePHprime{p}}{\upredicatePH{p}} \in \interfacesPHprime$ by Definition~\ref{def:theory_modular_encoding}, so $\domPHprime{p}\bigl(\semantics{\upredicatePHprime{p}}{\modelBVIA}{n}\bigr) = \extsemantics{\upredicatePH{p}}{\modelBVIA}{n+1} \subseteq \semantics{\upredicatePH{p}}{\modelBVIA}{n+1}$ for every $n$ (Eq.~\ref{eq:bounded_multi_semantics_external}), and the claim follows by taking the union over all bounds.
    Otherwise, $\locsemantics{\upredicatePHprime{p}}{\modelBVIA}{} = \emptyset$, so $\semantics{\upredicatePHprime{p}}{\modelBVIA}{} = \extsemantics{\upredicatePHprime{p}}{\modelBVIA}{} \subseteq \domPH{p}\bigl(\semantics{\upredicatePH{p}}{\modelBVIA}{}\bigr)$ by Eq.~\ref{eq:bounded_multi_semantics_external}, and applying $\domPHprime{p}$ yields the claim, since the composition $\domPHprime{p} \circ \domPH{p}$ is the identity (on $\domain^{\modelBV}_{p}$ when $\PH = \BVshort$, and on $\semantics{\upredicateIA{p}}{\modelBVIA}{}$ when $\PH = \IAshort$, using Lemma~\ref{lemma:well-defined}).
\end{itemize}

For the inclusion $\supseteq$, we show by induction on $n$ that $\semantics{\upredicateBV{p}}{\modelBVIA}{n} \subseteq \semantics{p}{\modelBV}{}$ and $\domIABV{p}\bigl(\semantics{\upredicateIA{p}}{\modelBVIA}{n}\bigr) \subseteq \semantics{p}{\modelBV}{}$ for every $p$ occurring in $\chcs$.
For the external components,
$\extsemantics{\upredicateBV{p}}{\modelBVIA}{n} \subseteq \domIABV{p}\bigl(\semantics{\upredicateIA{p}}{\modelBVIA}{n-1}\bigr) \subseteq \semantics{p}{\modelBV}{}$
and
$\domIABV{p}\bigl(\extsemantics{\upredicateIA{p}}{\modelBVIA}{n}\bigr) \subseteq \domIABV{p}\bigl(\domBVIA{p}\bigl(\semantics{\upredicateBV{p}}{\modelBVIA}{n-1}\bigr)\bigr) = \semantics{\upredicateBV{p}}{\modelBVIA}{n-1} \subseteq \semantics{p}{\modelBV}{}$
by Eq.~\ref{eq:bounded_multi_semantics_external} and the induction hypothesis.
For the local components, consider the contribution of a CHC with head predicate $p$.
If the CHC is $\upredicateBV{\chc} \in \BVfragment$ for some $\chc \in \Delta_{\BVshort}$, then by (A), its contribution equals $\semantics{\exists (\bar{x} \setminus \bar{x}_p).\, \body}{\modelBV\{p_j \mapsto R_j\}}{}$ with $R_j = \semantics{\upredicateBV{p}_j}{\modelBVIA}{n-1}$.
If the CHC is $\upredicateIA{\chc} \in \IAfragment$ for some $\chc \in \Delta_{\IAshort}$, then by (A) with $R_j = \domIABV{p_j}\bigl(\semantics{\upredicateIA{p}_j}{\modelBVIA}{n-1}\bigr)$ (well-defined by Lemma~\ref{lemma:well-defined}, and satisfying $\domBVIA{p_j}(R_j) = \semantics{\upredicateIA{p}_j}{\modelBVIA}{n-1}$), the image under $\domIABV{p}$ of its contribution equals $\semantics{\exists (\bar{x} \setminus \bar{x}_p).\, \body}{\modelBV\{p_j \mapsto R_j\}}{}$.
In both cases, $R_j \subseteq \semantics{p_j}{\modelBV}{}$ by the induction hypothesis, so by monotonicity, the resulting set is contained in $\semantics{\exists (\bar{x} \setminus \bar{x}_p).\, \body}{\modelBV\{p_j \mapsto \semantics{p_j}{\modelBV}{}\}}{}$, which is contained in $\semantics{p}{\modelBV}{}$ since the expansion $\modelBV\{p \mapsto \semantics{p}{\modelBV}{}\}_{p \in \upredicates}$ satisfies all rules in $\chcs$ (Section~\ref{subsec:chcs}).
The base case $n = 0$ is subsumed, as facts contain no predicate applications.

For the inclusion $\subseteq$, we show by induction on $n$ that $\semantics{p}{\modelBV}{n} \subseteq \semantics{\upredicateBV{p}}{\modelBVIA}{} \cup \domIABV{p}\bigl(\semantics{\upredicateIA{p}}{\modelBVIA}{}\bigr)$.
Consider the contribution of a CHC $\chc \in \chcs$ with head $p(\bar{x}_p)$, in which each $p_j$ is interpreted as $\semantics{p_j}{\modelBV}{n-1}$.
If $\chc \in \Delta_{\BVshort}$, then each $\upredicateBV{p}_j$ occurs in the body of $\upredicateBV{\chc} \in \BVfragment$, so by (B) and the induction hypothesis, $\semantics{p_j}{\modelBV}{n-1} \subseteq \semantics{\upredicateBV{p}_j}{\modelBVIA}{}$.
Since $\domain^{\modelBV}_{p_j}$ is finite, there is a common bound $m$ with $\semantics{p_j}{\modelBV}{n-1} \subseteq \semantics{\upredicateBV{p}_j}{\modelBVIA}{m}$ for all $j$, and by (A) and monotonicity, the contribution of $\chc$ is contained in that of $\upredicateBV{\chc}$ to $\locsemantics{\upredicateBV{p}}{\modelBVIA}{m+1}$, hence in $\semantics{\upredicateBV{p}}{\modelBVIA}{}$.
If $\chc \in \Delta_{\IAshort}$, then each $\upredicateIA{p}_j$ occurs in the body of $\upredicateIA{\chc} \in \IAfragment$, so by (B), $\domBVIA{p_j}\bigl(\semantics{\upredicateBV{p}_j}{\modelBVIA}{}\bigr) \subseteq \semantics{\upredicateIA{p}_j}{\modelBVIA}{}$, and together with the induction hypothesis and Lemma~\ref{lemma:well-defined}, $\domBVIA{p_j}\bigl(\semantics{p_j}{\modelBV}{n-1}\bigr) \subseteq \semantics{\upredicateIA{p}_j}{\modelBVIA}{}$.
By finiteness, there is a common bound $m$ with $\domBVIA{p_j}\bigl(\semantics{p_j}{\modelBV}{n-1}\bigr) \subseteq \semantics{\upredicateIA{p}_j}{\modelBVIA}{m}$ for all $j$.
By (A) with $R_j = \domIABV{p_j}\bigl(\semantics{\upredicateIA{p}_j}{\modelBVIA}{m}\bigr)$ and monotonicity (note $\semantics{p_j}{\modelBV}{n-1} \subseteq R_j$), the contribution of $\chc$ is contained in the image under $\domIABV{p}$ of the contribution of $\upredicateIA{\chc}$ to $\locsemantics{\upredicateIA{p}}{\modelBVIA}{m+1}$, hence in $\domIABV{p}\bigl(\semantics{\upredicateIA{p}}{\modelBVIA}{}\bigr)$.
\qed

\paragraph{\bf Proof of Theorem~\ref{theorem:semantic_preserving_encoding}.}
By Definition~\ref{def:theory_modular_encoding}, the queries of $\BVfragment$ and $\IAfragment$ are exactly the copies $\upredicateBV{\chc}$ and $\upredicateIA{\chc}$ of the queries $\chc \in \Delta_{\BVshort}$ and $\chc \in \Delta_{\IAshort}$, respectively.
By Definitions~\ref{def:satisfiability} and~\ref{def:multi_satisfiability}, it thus suffices to show that, for every query $\chc = \forall \bar{x}.\, \body[\bar{x}] \rightarrow \bot$ in $\chcs$,
\begin{displaymath}
    \modelBV\{p \mapsto \semantics{p}{\modelBV}{}\}_{p \in \upredicates} \models \chc
    \quad\text{iff}\quad
    \modelPH\{\upredicatePH{p} \mapsto \semantics{\upredicatePH{p}}{\modelBVIA}{}\}_{\upredicatePH{p} \in \upredicatesPH} \models \upredicatePH{\chc}
\end{displaymath}
where $\PH = \BVshort$ if $\chc \in \Delta_{\BVshort}$ and $\PH = \IAshort$ if $\chc \in \Delta_{\IAshort}$.
Let $\chc \in \Delta_{\BVshort}$.
Each $\upredicateBV{p}_j$ occurs in the body of $\upredicateBV{\chc} \in \BVfragment$, so by (B) (established in the proof of Lemma~\ref{lemma3}), $\domIABV{p_j}\bigl(\semantics{\upredicateIA{p}_j}{\modelBVIA}{}\bigr) \subseteq \semantics{\upredicateBV{p}_j}{\modelBVIA}{}$, and Lemma~\ref{lemma3} yields $\semantics{p_j}{\modelBV}{} = \semantics{\upredicateBV{p}_j}{\modelBVIA}{}$.
Since $\upredicateBV{\chc}$ differs from $\chc$ only in the renaming of its uninterpreted predicates, the two expanded structures agree on $\chc$, and the equivalence follows.
Let $\chc \in \Delta_{\IAshort}$.
Each $\upredicateIA{p}_j$ occurs in the body of $\upredicateIA{\chc} \in \IAfragment$, so by (B), $\domBVIA{p_j}\bigl(\semantics{\upredicateBV{p}_j}{\modelBVIA}{}\bigr) \subseteq \semantics{\upredicateIA{p}_j}{\modelBVIA}{}$.
Applying $\domBVIA{p_j}$ to the equality of Lemma~\ref{lemma3} and using Lemma~\ref{lemma:well-defined}, we obtain $\domBVIA{p_j}\bigl(\semantics{p_j}{\modelBV}{}\bigr) = \semantics{\upredicateIA{p}_j}{\modelBVIA}{}$.
By (A) with $R_j = \semantics{p_j}{\modelBV}{}$ and $\bar{x}_h$ the empty tuple, the body of $\chc$ is satisfiable in $\modelBV\{p_j \mapsto \semantics{p_j}{\modelBV}{}\}$ iff the body of $\upredicateIA{\chc}$ is satisfiable in $\modelIA\{\upredicateIA{p}_j \mapsto \semantics{\upredicateIA{p}_j}{\modelBVIA}{}\}$, and the equivalence follows.
\qed

\paragraph{\bf Proof of Theorem~\ref{theorem:solution_preservation}.}
We say that a formula $\varphi' \in \qfIA$ with free variables $\bar{x}'$ \emph{corresponds to} a formula $\varphi \in \qfBV$ with free variables $\bar{x}$ if $\semantics{\varphi}{\modelBV}{} = \domIABV{\bar{x}}\bigl(\restrBV{\semantics{\varphi'}{\modelIA}{}}\bigr)$, i.e., if $\varphi'$ preserves the semantics of $\varphi$ in the sense of Definition~\ref{def:theory_transformers}.
In particular, $\formulaBVIA(\varphi)$ corresponds to $\varphi$ and $\varphi'$ corresponds to $\formulaIABV(\varphi')$ (Eqs.~\ref{eq:theory_transformer_BV_IA} and~\ref{eq:theory_transformer_IA_BV}).
Moreover, $\bot$ corresponds to $\bot$, and $\bigwedge_{x' \in \bar{x}'} \restrBV{x'}$ corresponds to $\top$ by the definition of the syntactic domain restriction.
When formulas with free variables among a tuple $\bar{x}$ are combined, their semantics are taken over $\bar{x}$.
We use the following two observations.
\begin{itemize}
    \item[(O1)] Correspondence is preserved under conjunction and under substitution along the variable mapping.
    This holds since domain transformers are defined per variable, where the consistency requirement guarantees that substituted variables carry the same domain transformers.
    In addition, by Lemma~\ref{lemma:syntactic_restriction}, conjoining syntactic domain restrictions to $\varphi'$ preserves correspondence.
    \item[(O2)] Let $\varphi'_1, \varphi'_2$ correspond to $\varphi_1, \varphi_2$, respectively, with free variables among $\bar{x}$ (resp.\ $\bar{x}'$).
    Then $\forall\bar{x}.\, \varphi_1 \rightarrow \varphi_2$ is valid modulo $\theoryBV$ iff $\forall\bar{x}'.\, \bigl(\varphi'_1 \wedge \bigwedge_{x' \in \bar{x}'} \restrBV{x'}\bigr) \rightarrow \varphi'_2$ is valid modulo $\theoryIA$.
    Indeed, the former holds iff $\semantics{\varphi_1}{\modelBV}{} \subseteq \semantics{\varphi_2}{\modelBV}{}$, iff $\restrBV{\semantics{\varphi'_1}{\modelIA}{}} \subseteq \restrBV{\semantics{\varphi'_2}{\modelIA}{}}$ by correspondence and the injectivity of $\domBVIA{\bar{x}}$, iff $\semantics{\varphi'_1 \wedge \bigwedge_{x' \in \bar{x}'} \restrBV{x'}}{\modelIA}{} \subseteq \semantics{\varphi'_2}{\modelIA}{}$ by Lemma~\ref{lemma:syntactic_restriction}. %(for the last step, intersect both sides with the image of $\domBVIA{\bar{x}}$).
\end{itemize}

\smallskip
\noindent
($\Rightarrow$)
Let $\interpretation$ be a $\signatureBV$-solution for $\chcs$.
Define a $\signatureBVIA$-interpretation $\interpretationBVIA = (\interpretationBV, \interpretationIA)$ as follows.
For each uninterpreted predicate $p \in \upredicates$,
\begin{displaymath}
    \interpretationBV(\upredicateBV{p})[\bar{y}_{\upredicateBV{p}}]
    \triangleq
    \interpretation(p)[\bar{y}_{\upredicateBV{p}}]
    \quad\text{and}\quad
    \interpretationIA(\upredicateIA{p})[\bar{y}'_{\upredicateBV{p}}]
    \triangleq
    \formulaBVIA(\interpretation(p))[\bar{y}'_{\upredicateBV{p}}]
    \;\wedge\;
    \bigwedge_{y' \in \bar{y}'_{\upredicateBV{p}}} \restrBV{y'}
\end{displaymath}
We show that $\interpretationBVIA$ satisfies both conditions of Definition~\ref{def:multi_solution}.
Every CHC $\upredicateBV{\chc} \in \BVfragment$ is obtained from some $\chc \in \Delta_{\BVshort}$ by replacing each uninterpreted predicate $p$ with $\upredicateBV{p}$, so $\interpretationBV(\upredicateBV{\chc})$ coincides with $\interpretation(\chc)$ and is valid modulo $\theoryBV$.

Next, let $\upredicateIA{\chc} \in \IAfragment$ be obtained from $\chc = \forall \bar{x}.\, \bigl( \varphi \wedge \bigwedge_j p_j(\bar{x}_{p_j}) \bigr) \rightarrow H[\bar{x}]$ in $\Delta_{\IAshort}$ as in Definition~\ref{def:theory_modular_encoding}.
By (O1), the body of $\interpretationIA(\upredicateIA{\chc})$ corresponds to the body of $\interpretation(\chc)$: the constraint $\formulaBVIA(\varphi)$ corresponds to $\varphi$, each conjunct $\interpretationIA(\upredicateIA{p}_j)[\bar{y}'_{\upredicateBV{p_j}} \mapsto \bar{x}'_{p_j}]$ corresponds to $\interpretation(p_j)[\bar{y}_{\upredicateBV{p_j}} \mapsto \bar{x}_{p_j}]$, and the explicit restrictions correspond to $\top$.
Likewise, the head of $\interpretationIA(\upredicateIA{\chc})$ corresponds to the head of $\interpretation(\chc)$.
Moreover, the body of $\interpretationIA(\upredicateIA{\chc})$ entails $\bigwedge_{x' \in \bar{x}'} \restrBV{x'}$ in $\modelIA$, since the restrictions occurring in the conjuncts $\interpretationIA(\upredicateIA{p}_j)[\cdot]$ cover $\bigcup_j \bar{x}'_{p_j}$ and the explicit restrictions cover $\bar{y} = \bar{x} \setminus (\bigcup_j \bar{x}_{p_j})$.
Conjoining $\bigwedge_{x' \in \bar{x}'} \restrBV{x'}$ to the body of $\interpretationIA(\upredicateIA{\chc})$ therefore results in an equivalent formula, and by (O2), the validity of $\interpretation(\chc)$ modulo $\theoryBV$ implies the validity of $\interpretationIA(\upredicateIA{\chc})$ modulo $\theoryIA$.

Finally, for every $p \in \upredicates$,
\begin{displaymath}
    \semantics{\formulaIABV(\interpretationIA(\upredicateIA{p}))}{\modelBV}{}
    =
    \domIABV{\bar{y}_{\upredicateBV{p}}}\bigl(\restrBV{\semantics{\interpretationIA(\upredicateIA{p})}{\modelIA}{}}\bigr)
    =
    \domIABV{\bar{y}_{\upredicateBV{p}}}\bigl(\restrBV{\semantics{\formulaBVIA(\interpretation(p))}{\modelIA}{}}\bigr)
    =
    \semantics{\interpretation(p)}{\modelBV}{}
\end{displaymath}
by Eq.~\ref{eq:theory_transformer_IA_BV}, Lemma~\ref{lemma:syntactic_restriction}, and Eq.~\ref{eq:theory_transformer_BV_IA}, respectively.
Since $\interpretation(p) = \interpretationBV(\upredicateBV{p})$, both $\forall \bar{y}_{\upredicateBV{p}}.\, \interpretationBV(\upredicateBV{p}) \rightarrow \formulaIABV(\interpretationIA(\upredicateIA{p}))$ and its converse are valid modulo $\theoryBV$, so $\interpretationBVIA(\theta)$ is valid modulo $\theoryBV$ for every interface constraint $\theta \in \interfaces$.

\smallskip
\noindent
($\Leftarrow$)
Let $\interpretationBVIA = (\interpretationBV, \interpretationIA)$ be a $\signatureBVIA$-solution for $\encoding$.
Define a $\signatureBV$-interpretation $\interpretation$ as follows.
For each uninterpreted predicate $p \in \upredicates$,
\begin{displaymath}
    \interpretation(p)[\bar{y}_{\upredicateBV{p}}]
    \;\triangleq\;
    \varphi^{p}_{\BVshort} \;\vee\; \varphi^{p}_{\IAshort}
\end{displaymath}
where
\begin{align*}
    \varphi^{p}_{\BVshort} &\triangleq
    \begin{cases}
        \interpretationBV(\upredicateBV{p})[\bar{y}_{\upredicateBV{p}}] & \text{if $\upredicateBV{p}$ is the head predicate of some CHC in $\BVfragment$} \\
        \bot & \text{otherwise}
    \end{cases}
    \\
    \varphi^{p}_{\IAshort} &\triangleq
    \begin{cases}
        \formulaIABV\bigl(\interpretationIA(\upredicateIA{p})\bigr)[\bar{y}_{\upredicateBV{p}}] & \text{if $\upredicateIA{p}$ is the head predicate of some CHC in $\IAfragment$} \\
        \bot & \text{otherwise}
    \end{cases}
\end{align*}
Note that $\formulaIABV(\interpretationIA(\upredicateIA{p}))$ is a $\signatureBV$-formula with free variables in $\bar{y}_{\upredicateBV{p}}$, since $\bar{y}_{\upredicateIA{p}} = \bar{y}'_{\upredicateBV{p}}$, so $\interpretation$ is well-defined.
%Restricting each disjunct to head-side copies is essential: if $\upredicateIA{p}$ is not the head predicate of any CHC in $\IAfragment$ then, by Definition~\ref{def:theory_modular_encoding}, no interface constraint in $\interfaces$ has source $\upredicateIA{p}$, so $\interpretationIA(\upredicateIA{p})$ is not bounded in terms of $\interpretationBV(\upredicateBV{p})$, and including it in the disjunction could enlarge $\interpretation(p)$ on body occurrences of $p$ (symmetrically for $\upredicateBV{p}$).

We first establish the following claim.
For every $p \in \upredicates$:
\begin{itemize}
    \item[(C1)] If $\upredicateBV{p}$ occurs in the body of some CHC in $\BVfragment$, then $\forall \bar{y}_{\upredicateBV{p}}.\, \interpretation(p) \rightarrow \interpretationBV(\upredicateBV{p})$ is valid modulo $\theoryBV$.
    \item[(C2)] If $\upredicateIA{p}$ occurs in the body of some CHC in $\IAfragment$, then $\forall \bar{y}_{\upredicateBV{p}}.\, \interpretation(p) \rightarrow \formulaIABV\bigl(\interpretationIA(\upredicateIA{p})\bigr)$ is valid modulo $\theoryBV$.
\end{itemize}
For (C1), the disjunct $\varphi^{p}_{\BVshort}$ is either $\interpretationBV(\upredicateBV{p})$ or $\bot$.
If $\varphi^{p}_{\IAshort} \neq \bot$, then $\upredicateIA{p}$ is the head predicate of some CHC in $\IAfragment$ and, by the assumption of (C1), $\upredicateBV{p}$ occurs in the body of some CHC in $\BVfragment$.
By Definition~\ref{def:theory_modular_encoding}, $\interface{\upredicateIA{p}}{\upredicateBV{p}} \in \interfacesIABV$, and interface constraint satisfaction (Definition~\ref{def:multi_solution}) yields the validity modulo $\theoryBV$ of $\varphi^{p}_{\IAshort} \rightarrow \interpretationBV(\upredicateBV{p})$.
The proof of (C2) is symmetric, using $\interface{\upredicateBV{p}}{\upredicateIA{p}} \in \interfacesBVIA$.
In addition, by construction, if $\upredicateBV{p}$ (resp.\ $\upredicateIA{p}$) is the head predicate of some CHC in its fragment, then $\interpretationBV(\upredicateBV{p}) \rightarrow \interpretation(p)$ (resp.\ $\formulaIABV(\interpretationIA(\upredicateIA{p})) \rightarrow \interpretation(p)$) is valid modulo $\theoryBV$.
We refer to these implications as the \emph{head bounds}.
Being validities of universally quantified implications, (C1), (C2), and the head bounds also hold under any substitution of $\bar{y}_{\upredicateBV{p}}$.

We now show that $\interpretation(\chc)$ is valid modulo $\theoryBV$ for every $\chc = \forall \bar{x}.\, \bigl( \varphi \wedge \bigwedge_j p_j(\bar{x}_{p_j}) \bigr) \rightarrow H[\bar{x}]$ in $\chcs$, i.e., that $\interpretation$ is a $\signatureBV$-solution for $\chcs$ (Definition~\ref{def:solution}).

\emph{Case $\chc \in \Delta_{\BVshort}$.}
Let $\upredicateBV{\chc} \in \BVfragment$ be the copy of $\chc$.
Each $\upredicateBV{p}_j$ occurs in the body of $\upredicateBV{\chc}$, so by (C1), the body of $\interpretation(\chc)$ implies the body of $\interpretationBV(\upredicateBV{\chc})$ modulo $\theoryBV$.
By local satisfaction, the latter implies $\bot$ if $H = \bot$, and $\interpretationBV(\upredicateBV{h})[\bar{y}_{\upredicateBV{h}} \mapsto \bar{x}_h]$ if $H = h(\bar{x}_h)$.
In the latter case, $\upredicateBV{h}$ is the head predicate of $\upredicateBV{\chc} \in \BVfragment$, so the head bound gives $\interpretationBV(\upredicateBV{h}) \rightarrow \interpretation(h)$.
Chaining the implications shows that $\interpretation(\chc)$ is valid modulo $\theoryBV$.

\emph{Case $\chc \in \Delta_{\IAshort}$.}
Let $\upredicateIA{\chc} \in \IAfragment$ be the translation of $\chc$ per Definition~\ref{def:theory_modular_encoding}.
Each $\upredicateIA{p}_j$ occurs in the body of $\upredicateIA{\chc}$, so by (C2), the body of $\interpretation(\chc)$ implies, modulo $\theoryBV$, the $\signatureBV$-formula
\begin{equation}\label{eq:solution_preservation_body}
    \varphi \;\wedge\; \bigwedge_j \formulaIABV\bigl(\interpretationIA(\upredicateIA{p}_j)\bigr)[\bar{y}_{\upredicateBV{p_j}} \mapsto \bar{x}_{p_j}]
\end{equation}
By (O1), the body of $\interpretationIA(\upredicateIA{\chc})$ corresponds to Eq.~\ref{eq:solution_preservation_body}, and the head of $\interpretationIA(\upredicateIA{\chc})$ corresponds to $\formulaIABV(\interpretationIA(\upredicateIA{h}))[\bar{y}_{\upredicateBV{h}} \mapsto \bar{x}_h]$ if $H = h(\bar{x}_h)$, and to $\bot$ if $H = \bot$.
By local satisfaction, $\interpretationIA(\upredicateIA{\chc})$ is valid modulo $\theoryIA$, and hence so is the formula obtained from it by conjoining $\bigwedge_{x' \in \bar{x}'} \restrBV{x'}$ to its body.
By (O2), Eq.~\ref{eq:solution_preservation_body} implies the corresponding head modulo $\theoryBV$.
If $H = \bot$, this shows that $\interpretation(\chc)$ is valid modulo $\theoryBV$.
Otherwise, $\upredicateIA{h}$ is the head predicate of $\upredicateIA{\chc} \in \IAfragment$, so the head bound gives $\formulaIABV(\interpretationIA(\upredicateIA{h})) \rightarrow \interpretation(h)$, and chaining the implications shows that $\interpretation(\chc)$ is valid modulo $\theoryBV$. \qed

\paragraph{\bf Proof of Theorem~\ref{theorem:soundness}.}
The proof is by induction on the length of the transition sequence.
We maintain the following invariant.
For every $\icformula{\upredicatePH{p}}{\psi} \in \ulemmas$, the formula $\psi$ under-approximates the external semantics of $\upredicatePH{p}$:
if $\PH=\BVshort$, then
$\semantics{\psi}{\modelBV}{} \subseteq \extsemantics{\upredicateBV{p}}{\modelBVIA}{}$,
and if $\PH=\IAshort$, then
$\restrBV{\semantics{\psi}{\modelIA}{}} \subseteq \extsemantics{\upredicateIA{p}}{\modelBVIA}{}$.
Dually, for every $\icformula{\upredicatePH{p}}{\psi} \in \olemmas$, the formula $\psi$ over-approximates the external semantics of $\upredicatePH{p}$:
if $\PH=\BVshort$, then
$\extsemantics{\upredicateBV{p}}{\modelBVIA}{} \subseteq \semantics{\psi}{\modelBV}{}$,
and if $\PH=\IAshort$, then
$\extsemantics{\upredicateIA{p}}{\modelBVIA}{} \subseteq \restrBV{\semantics{\psi}{\modelIA}{}}$.

The invariant holds in the initial state (rule \textsc{Init}), since both lemma sets are empty.
Rule \textsc{Query} preserves it trivially, as it does not modify $\ulemmas$ or $\olemmas$.
Consider rule \textsc{PosAnswer}.
Let $\refutation$ be the refutation of $\fragmentPHunder \cup \Query(\upredicatePH{p},\phi)$ returned by the oracle, and let $\langle \upredicatePH{p},\psi\rangle$ be the child of its root.
By construction, $\fragmentPHunder$ consists of CHCs from $\fragmentPH$ together with facts induced by $\ulemmas$.
CHCs in $\fragmentPH$ are part of the original encoding $\encoding$, and facts induced by $\ulemmas$ under-approximate the corresponding external semantics by the induction hypothesis.
Hence $\fragmentPHunder$ under-approximates the behavior of $\encoding$, in the sense that for every uninterpreted predicate $\upredicatePH{r}$ in $\fragmentPH$ (and in particular for $\upredicatePH{p}$), its semantics in $\fragmentPHunder$ is contained in its semantics in $\encoding$.
Therefore, $\refutation$ corresponds to a feasible (partial) derivation in $\encoding$, and $\psi$ under-approximates the semantics of $\upredicatePH{p}$ in $\encoding$.
Finally, by Definitions~\ref{def:bounded_multi_semantics} and~\ref{def:theory_transformers}, translating $\psi$ across the corresponding interface constraint yields a formula that under-approximates the external semantics of $\upredicatePHprime{p}$.
Thus the invariant is preserved for $\ulemmas$.

Rule \textsc{NegAnswer} is dual.
Let $\interpretationPH$ be the solution of $\fragmentPHover \cup \Query(\upredicatePH{p},\phi)$ returned by the oracle.
By construction, $\fragmentPHover$ consists of CHCs from $\fragmentPH$ together with facts induced by $\olemmas$.
CHCs in $\fragmentPH$ are part of the original encoding $\encoding$, and facts induced by $\olemmas$ over-approximate the corresponding external semantics by the induction hypothesis.
Hence $\fragmentPHover$ over-approximates the behavior of $\encoding$, in the sense that for every uninterpreted predicate $\upredicatePH{r}$ in $\fragmentPH$ (and in particular for $\upredicatePH{p}$), its semantics in $\encoding$ is contained in its semantics in $\fragmentPHover$.
It follows that $\interpretationPH(\upredicatePH{p})$ over-approximates the semantics of $\upredicatePH{p}$ in $\encoding$.
Finally, by Definitions~\ref{def:bounded_multi_semantics} and~\ref{def:theory_transformers}, translating $\interpretationPH(\upredicatePH{p})$ across the corresponding interface constraint yields a formula that over-approximates the external semantics of $\upredicatePHprime{p}$.
Thus the invariant is preserved for $\olemmas$.

It follows from the invariant that, throughout the run, $\fragmentPHunder$ under-approximates the semantics of the corresponding fragment in the theory-modular encoding, while $\fragmentPHover$ over-approximates it.
Hence, the soundness of rule \textsc{Unsat} follows by the same reasoning as for rule \textsc{PosAnswer}: a refutation of
$\BVfragmentunder \cup \Query(\upredicateBV{q},\phi_q)$ corresponds to a feasible derivation in $\encoding$, and thus witnesses a violation of the query.
Dually, the soundness of rule \textsc{Sat} follows by the same reasoning as for rule \textsc{NegAnswer}: a solution of
$\BVfragmentover \cup \Query(\upredicateBV{q},\phi_q)$ over-approximates the semantics in $\encoding$ and proves that no query violation is possible.
This concludes the proof.

%% file: appendix/benchmarks.tex
\clearpage
\section{Benchmarks}\label{appendix:benchmarks}

\input{benchmarks/abs-ge}
\input{benchmarks/abs-sum}
\input{benchmarks/cond_negate}
\input{benchmarks/cond_negate_diff}
\input{benchmarks/max-inv}
\input{benchmarks/opposite_signs_diff}
\input{benchmarks/swap}
\input{benchmarks/swap_sum}
\input{benchmarks/turn_off_rm1}
\input{benchmarks/turn_on_lsb}

%% file: benchmarks/abs-ge.tex
\begin{figure}[h]
\centering

% Left column: Panel 1
\begin{minipage}[h]{0.32\textwidth}
\begin{lstlisting}[style=compactC]
void f(int x) {
  assume(x!=INT_MIN);
  int y,i;
  if(x>=0) {
    y=(1|0)*x;
  } else {
    y=(1|-1)*x;
  }
  i=0;
  while(i<y) {
    i++;
  }
  assert(x<=i);
}
\end{lstlisting}
\centering (a)
\end{minipage}
\hfill
% Right column: Panels 2 and 3 stacked vertically
\begin{minipage}[h]{0.67\textwidth}

% Top: Panel 2
\begin{minipage}[h]{\textwidth}
\[
{\footnotesize
\begin{aligned}
\chcs = \{ \;
    & ((x \bvsge 0, y = (1 \bvor 0) \bvmul x) \vee (x \bvslt 0, y = (1 \bvor -1) \bvmul x)), \\
    & \hspace*{2.85cm} x \neq (1 \bvshl (n \bvsub 1)) \; \rightarrow \; p(x,y,0), \\
    & p(x,y,i),\; i \bvslt y \; \rightarrow \; p(x,y,i \bvadd 1), \\
    & p(x,y,i),\; \neg(i \bvslt y) \; \rightarrow \; q(x,i), \\
    & q(x,i),\; \neg(x \bvsle i) \; \rightarrow \; \bot
\; \}
\end{aligned}
}%
\]
\centering (b)
\end{minipage}

\vspace{0.5em} % small vertical gap

% Bottom: Panel 3
\begin{minipage}[h]{\textwidth}
\[
{\footnotesize 
\begin{aligned}
\BVfragment = \{ \;
    & ((x \bvsge 0, y = (1 \bvor 0) \bvmul x) \vee (x \bvslt 0, y = (1 \bvor -1) \bvmul x)), \\
    & \hspace*{2.85cm} x \neq (1 \bvshl (n \bvsub 1)) \; \rightarrow \; \upredicateBV{p}(x,y,0)
\; \} \\[1ex]
\IAfragment = \{ \;
    & \upredicateIA{p}(x',y',i'),\; i' < y' \; \rightarrow \; \upredicateIA{p}(x',y',i'+1), \\
    & \upredicateIA{p}(x',y',i'),\; \neg(i' < y') \; \rightarrow \; \upredicateIA{q}(x',i'), \\
    & \upredicateIA{q}(x',i'),\; \neg(x' \leq i') \; \rightarrow \; \bot
\; \} \\[1ex]
\interfaces = \{\; &\interface{\upredicateBV{p}}{\upredicateIA{p}}\; \}
\end{aligned}
}%
\]
\centering (c)
\end{minipage}

\end{minipage}

\caption{The \textsf{abs-ge} benchmark.
$y$ is initialized to the absolute value of $x$.
All variables and constants in $\chcs$ have sort $\sortBV{n}$ for some $n \in \mathbb{N}^+$.
}
\end{figure}

%% file: benchmarks/abs-sum.tex
\begin{figure}[t]
\centering

% Left column: Panel 1
\begin{minipage}[h]{0.35\textwidth}
%\vspace{9cm} 
\begin{lstlisting}[style=compactC]
void f(int x, int y) {
  assume(x>0);
  assume(y>0);
  int y'=-y,a=0,i=0;
  int b,c;
  while(i<x) {
    a++;
    i++;
  }
  i=0;
  while(i>y') {
    a--;
    i--;
  }
  if(x>=0) {
    b=(1|0)*x;
  } else {
    b=(1|-1)*x;
  }
  if(y'>=0) {
    c=(1|0)*y';
  } else {
    c=(1|-1)*y';
  }
  if(a>=0) {
    assert(b>=c);
  }
}
\end{lstlisting}
\centering (a)
\end{minipage}
\hfill
% Right column: Panels 2 and 3 stacked vertically
\begin{minipage}[h]{0.60\textwidth}

% Top: Panel 2
\begin{minipage}[h]{\textwidth}
\[
{\footnotesize
\begin{aligned}
\chcs = \{ \;
    & x \bvsgt 0,\; y \bvsgt 0 \; \rightarrow \; p(x,\bvneg y,0,0), \\
    & p(x,y,a,i),\; i \bvslt x \; \rightarrow \; p(x,y,a \bvadd 1,i \bvadd 1), \\
    & p(x,y,a,i),\; \neg(i \bvslt x) \; \rightarrow \; q(x,y,a,0), \\
    & q(x,y,a,i),\; i \bvsgt y \; \rightarrow \; q(x,y,a \bvsub 1,i \bvsub 1), \\
    & q(x,y,a,i),\; \neg(i \bvsgt y) \; \rightarrow \; r(x,y,a),\\
    & r(x,y,a),\; (a \bvsge 0),\; \\
    & \hspace*{0.5cm} ((x \bvsge 0, b = (1 \bvor 0) \bvmul x) \vee (x \bvslt 0, b = (1 \bvor -1) \bvmul x)), \\
    & \hspace*{0.5cm}((y \bvsge 0, c = (1 \bvor 0) \bvmul y) \vee (y \bvslt 0, c = (1 \bvor -1) \bvmul y)), \\
    & \hspace*{5.55cm}\neg(b \bvsge c) \; \rightarrow \; \bot
\; \}
\end{aligned}
}%
\]
\centering (b)
\end{minipage}

\vspace{0.5em} % small vertical gap

% Bottom: Panel 3
\begin{minipage}[h]{\textwidth}
\[
{\footnotesize 
\begin{aligned}
\IAfragment = \{ \;
    & x' > 0,\; y' > 0,\; \restrBV{x'},\; \restrBV{y'} \; \rightarrow \; \upredicateIA{p}(x',-y',0,0), \\
    & \upredicateIA{p}(x',y',a',i'),\; i' < x' \; \rightarrow \; \upredicateIA{p}(x',y',a'+1,i'+1), \\
    & \upredicateIA{p}(x',y',a',i'),\; \neg(i' < x') \; \rightarrow \; \upredicateIA{q}(x',y',a',0), \\
    & \upredicateIA{q}(x',y',a',i'),\; i' > y' \; \rightarrow \; \upredicateIA{q}(x',y',a'-1,i'-1), \\
    & \upredicateIA{q}(x',y',a',i'),\; \neg(i' > y') \; \rightarrow \; \upredicateIA{r}(x',y',a')
\; \} \\[1ex]
\BVfragment = \{ \;
    & \upredicateBV{r}(x,y,a),\; (a \bvsge 0),\; \\
    & \hspace*{0.5cm} ((x \bvsge 0, b = (1 \bvor 0) \bvmul x) \vee (x \bvslt 0, b = (1 \bvor -1) \bvmul x)), \\
    & \hspace*{0.5cm}((y \bvsge 0, c = (1 \bvor 0) \bvmul y) \vee (y \bvslt 0, c = (1 \bvor -1) \bvmul y)), \\
    & \hspace*{5.55cm}\neg(b \bvsge c) \; \rightarrow \; \bot
\; \} \\[1ex]
\interfaces = \{\; &\interface{\upredicateIA{r}}{\upredicateBV{r}}\; \}
\end{aligned}
}%
\]
\centering (c)
\end{minipage}

\end{minipage}

\caption{The \textsf{abs-sum} benchmark. 
All variables and constants in $\chcs$ have sort $\sortBV{n}$ for some $n \in \mathbb{N}^+$.
Upon termination of the function, the program variables satisfy: \texttt{a = x - y}, \texttt{b = abs(x)}, and \texttt{c = abs(y)}.
}

\end{figure}

%% file: benchmarks/cond_negate.tex
\begin{figure}[t]
\centering

% Left column: Panel 1
\begin{minipage}[h]{0.32\textwidth}
\begin{lstlisting}[style=compactC]
void f(int x,int y) {
  assume(x>y);
  assume(y>0);
  int i=0,b;
  while(i<y) {
    i+=2;
  }
  if(i<=x) {
    b=1;
  } else {
    b=0;
  }
  assert((x^(-b))+b==-x);
}
\end{lstlisting}
\centering (a)
\end{minipage}
\hfill
% Right column: Panels 2 and 3 stacked vertically
\begin{minipage}[h]{0.67\textwidth}

% Top: Panel 2
\begin{minipage}[h]{\textwidth}
\[
{\footnotesize
\begin{aligned}
\chcs = \{ \;
    & x \bvsgt y,\; y \bvsgt 0 \; \rightarrow \; p(x,y,0), \\
    & p(x,y,i),\; i \bvslt y \; \rightarrow \; p(x,y,i \bvadd 2), \\
    & p(x,y,i),\; \neg(i \bvslt y) \; \rightarrow \; q(x,i), \\
    & q(x,i),\; ((i \bvsle x, b = 1) \vee ((i \bvsgt x, b = 0)), \\
    & \hspace*{1cm} \neg((x \bvxor (\bvneg b)) \bvadd b = \bvneg x)) \; \rightarrow \; \bot
\; \}
\end{aligned}
}%
\]
\centering (b)
\end{minipage}

\vspace{0.5em} % small vertical gap

% Bottom: Panel 3
\begin{minipage}[h]{\textwidth}
\[
{\footnotesize 
\begin{aligned}
\IAfragment = \{ \;
    & x' > y',\; y' > 0,\; \restrBV{x'},\; \restrBV{y'} \; \rightarrow \; \upredicateIA{p}(x',y',0), \\
    & \upredicateIA{p}(x',y',i'),\; i' < y' \; \rightarrow \; \upredicateIA{p}(x',y',i'+2), \\
    & \upredicateIA{p}(x',y',i'),\; \neg(i' < y') \; \rightarrow \; \upredicateIA{q}(x',i')
\; \} \\[1ex]
\BVfragment = \{ \;
    & \upredicateBV{q}(x,i),\; ((i \bvsle x, b = 1) \vee ((i \bvsgt x, b = 0)), \\
    & \hspace*{1cm} \neg((x \bvxor (\bvneg b)) \bvadd b = \bvneg x)) \; \rightarrow \; \bot
\; \} \\[1ex]
\interfaces = \{\; &\interface{\upredicateIA{q}}{\upredicateBV{q}}\; \}
\end{aligned}
}%
\]
\centering (c)
\end{minipage}

\end{minipage}

\caption{The \textsf{cond-neg} benchmark.
All variables and constants in $\chcs$ have sort $\sortBV{n}$ for some $n \in \mathbb{N}^+$.
The term $((x \bvxor (\bvneg b)) \bvadd b)$ evaluates to $\bvneg x$ if $b=1$ and to $x$ if $b=0$.}
\end{figure}

%% file: benchmarks/cond_negate_diff.tex
\begin{figure}[t]
\centering

% Left column: Panel 1
\begin{minipage}[h]{0.32\textwidth}
\begin{lstlisting}[style=compactC]
void f(int x,int y) {
  assume(x>y);
  assume(y>0);
  int i=0,a=0,b;
  while(i<x) {
    a--;
    i++;
  }
  i=0;
  while(i<y) {
    a++;
    i++;
  }
  if(a<0) {
    b=1;
  } else {
    b=0;
  }
  assert((a^(-b))+b==-a);
}
\end{lstlisting}
\centering (a)
\end{minipage}
\hfill
% Right column: Panels 2 and 3 stacked vertically
\begin{minipage}[h]{0.67\textwidth}

% Top: Panel 2
\begin{minipage}[h]{\textwidth}
\[
{\footnotesize
\begin{aligned}
\chcs = \{ \;
    & x \bvsgt y,\; y \bvsgt 0 \; \rightarrow \; p(x,y,0,0), \\
    & p(x,y,a,i),\; i \bvslt x \; \rightarrow \; p(x,y,a \bvsub 1,i \bvadd 1), \\
    & p(x,y,a,i),\; \neg(i \bvslt x) \; \rightarrow \; q(y,a,0), \\
    & q(y,a,i),\; i \bvslt y \; \rightarrow \; q(y,a \bvadd 1,i \bvadd 1), \\
    & q(y,a,i),\; \neg(i \bvslt y),\; a \bvslt 0 \; \rightarrow \; r(a,1), \\
    & q(y,a,i),\; \neg(i \bvslt y),\; \neg(a \bvslt 0) \; \rightarrow \; r(a,0), \\
    & r(a,b),\; \neg((a \bvxor (\bvneg b)) \bvadd b = \bvneg a) \; \rightarrow \; \bot
\; \}
\end{aligned}
}%
\]
\centering (b)
\end{minipage}

\vspace{0.5em} % small vertical gap

% Bottom: Panel 3
\begin{minipage}[h]{\textwidth}
\[
{\footnotesize 
\begin{aligned}
\IAfragment = \{ \;
    & x' > y',\; y' > 0,\; \restrBV{x'},\; \restrBV{y'} \; \rightarrow \; \upredicateIA{p}(x',y',0,0), \\
    & \upredicateIA{p}(x',y',a',i'),\; i' < x' \; \rightarrow \; \upredicateIA{p}(x',y',a'-1,i'+1), \\
    & \upredicateIA{p}(x',y',a',i'),\; \neg(i' < x') \; \rightarrow \; \upredicateIA{q}(y',a',0), \\
    & \upredicateIA{q}(y',a',i'),\; i' < y' \; \rightarrow \; \upredicateIA{q}(y',a'+1,i'+1), \\
    & \upredicateIA{q}(y',a',i'),\; \neg(i' < y'),\; a' < 0 \; \rightarrow \; \upredicateIA{r}(a',1), \\
    & \upredicateIA{q}(y',a',i'),\; \neg(i' < y'),\; \neg(a' < 0) \; \rightarrow \; \upredicateIA{r}(a',0)
\; \} \\[1ex]
\BVfragment = \{ \;
    & \upredicateBV{r}(a,b),\; \neg((a \bvxor (\bvneg b)) \bvadd b = \bvneg a) \; \rightarrow \; \bot
\; \} \\[1ex]
\interfaces = \{\; &\interface{\upredicateIA{r}}{\upredicateBV{r}}\; \}
\end{aligned}
}%
\]
\centering (c)
\end{minipage}

\end{minipage}

\caption{The \textsf{cond-neg-diff} benchmark.
All variables and constants in $\chcs$ have sort $\sortBV{n}$ for some $n \in \mathbb{N}^+$.
The term $((a \bvxor (\bvneg b)) \bvadd b)$ evaluates to $\bvneg a$ if $b=1$ and to $a$ if $b=0$.}
\end{figure}

%% file: benchmarks/max-inv.tex
\begin{figure}[h]
\centering

% Left column: Panel 1
\begin{minipage}[h]{0.36\textwidth}
%\vspace{3em} 
\begin{lstlisting}[style=compactC]
void f(unsigned x,unsigned y,unsigned z) {
  assume(x>y);
  assume(x-y>=z);
  unsigned a=x;
  unsigned i=0;
  unsigned max=0;
  while(i<z) {
    a--;
    i++;
  }
  if (a<y) {
    max=a^((a^y)&-1);
  } else {
    max=a^((a^y)&0);
  }
  assert(a==max);
}
\end{lstlisting}
\centering (a)
\end{minipage}
\hfill
% Right column: Panels 2 and 3 stacked vertically
\begin{minipage}[h]{0.62\textwidth}

% Top: Panel 2
\begin{minipage}[h]{\textwidth}
\[
{\footnotesize
\begin{aligned}
\chcs = \{ \;
& x \bvugt y, \; x \bvsub y \bvuge z \; \rightarrow \; p(y,z,x,0), \\
& p(y,z,a,i), \; i \bvult z \; \rightarrow \; p(y,z,a\bvsub 1, i\bvadd 1), \\
& p(y,z,a,i), \; \neg(i \bvult z) \; \rightarrow \; q(y,a), \\
& q(y,a), \; ((a \bvult y, b=-1) \vee (\neg(a \bvult y), b=0)), \\
& \hspace*{1.5cm} \neg(a = (a \bvxor ((a \bvxor y) \; \bvand \; b))) \; \rightarrow \; \bot
\; \}
\end{aligned}
}%
\]
\centering (b)
\end{minipage}

\vspace{0.5em} % small vertical gap

% Bottom: Panel 3
\begin{minipage}[h]{\textwidth}
\[
{\footnotesize 
\begin{aligned}
\IAfragment = \{ \;
  & x' > y', \; x' - y' \geq z', \; \restrBV{x'}, \; \restrBV{y'}, \; \restrBV{z'} \; \rightarrow \\
  & \hspace*{4.9cm} \upredicateIA{p}(y',z',x',0), \\
  & \upredicateIA{p}(y',z',a',i'), \; i' < z' \; \rightarrow \; \upredicateIA{p}(y',z',a'-1, i'+1), \\
  & \upredicateIA{p}(y',z',a',i'), \; \neg(i' < z') \; \rightarrow \; \upredicateIA{q}(y',a')
\; \} \\[1ex]
\BVfragment = \{ \;
  & \upredicateBV{q}(y,a), \; ((a \bvult y, b=-1) \vee (\neg(a \bvult y), b=0)), \\
  & \hspace*{1.5cm} \neg(a = (a \bvxor ((a \bvxor y) \; \bvand \; b))) \; \rightarrow \; \bot
\; \} \\[1ex]
\interfaces = \{\; &\interface{\upredicateIA{q}}{\upredicateBV{q}}\; \}
\end{aligned}
}%
\]
\centering (c)
\end{minipage}

\end{minipage}

\caption{The \textsf{max-inv} benchmark.
All variables and constants in $\chcs$ have sort $\sortBV{n}$ for some $n \in \mathbb{N}^+$.
The term $(a \bvxor ((a \bvxor y) \;\bvand\; b))$ evaluates to $\max(a,y)$, where $b=-1$ if $a \bvult y$ and $b=0$ otherwise.}
\vspace{-12pt}
\end{figure}

%% file: benchmarks/opposite_signs_diff.tex
\begin{figure}[t]
\centering

% Left column: Panel 1
\begin{minipage}[h]{0.34\textwidth}
\begin{lstlisting}[style=compactC]
void f(int x, int y) {
  assume(x>y);
  assume(y>0);
  int a=0,b=0,i=0;
  while(i<x) {
    a++;
    b--;
    i++;
  }
  i=0;
  while(i<y) {
    a--;
    b++;
    i++;
  }
  assert((a^b)<0);
}
\end{lstlisting}
\centering (a)
\end{minipage}
\hfill
% Right column: Panels 2 and 3 stacked vertically
\begin{minipage}[h]{0.60\textwidth}

% Top: Panel 2
\begin{minipage}[h]{\textwidth}
\[
{\footnotesize
\begin{aligned}
\chcs = \{ \;
    & x \bvsgt y,\; y \bvsgt 0 \; \rightarrow \; p(x,y,0,0,0), \\
    & p(x,y,a,b,i),\; i \bvslt x \; \rightarrow \; p(x,y,a \bvadd 1,b \bvsub 1, i \bvadd 1), \\ 
    & p(x,y,a,b,i),\; \neg(i \bvslt x) \; \rightarrow \; q(y,a,b,0), \\
    & q(y,a,b,i),\; i \bvslt y \; \rightarrow \; q(y,a \bvsub 1,b \bvadd 1,i \bvadd 1), \\
    & q(y,a,b,i),\; \neg(i \bvslt y) \; \rightarrow \; r(a,b), \\
    & r(a,b), \; \neg((a \bvxor b) \bvslt 0) \; \rightarrow \; \bot 
\; \}
\end{aligned}
}%
\]
\centering (b)
\end{minipage}

\vspace{0.5em} % small vertical gap

% Bottom: Panel 3
\begin{minipage}[h]{\textwidth}
\[
{\footnotesize 
\begin{aligned}
\IAfragment = \{ \;
    & x' > y',\; y' > 0,\; \restrBV{x'},\; \restrBV{y'} \; \rightarrow \; \upredicateIA{p}(x',y',0,0,0), \\
    & \upredicateIA{p}(x',y',a',b',i'),\; i' < x' \; \rightarrow \; \upredicateIA{p}(x',y',a'+1,b'-1,i'+1), \\ 
    & \upredicateIA{p}(x',y',a',b',i'),\; \neg(i' < x') \; \rightarrow \; \upredicateIA{q}(y',a',b',0), \\
    & \upredicateIA{q}(y',a',b',i'),\; i' < y' \; \rightarrow \; \upredicateIA{q}(y',a'-1,b'+1,i'+1), \\
    & \upredicateIA{q}(y',a',b',i'),\; \neg(i' < y') \; \rightarrow \; \upredicateIA{r}(a',b')
\; \} \\[1ex]
\BVfragment = \{ \;
  & \upredicateBV{r}(a,b), \; \neg((a \bvxor b) \bvslt 0) \; \rightarrow \; \bot
\; \} \\[1ex]
\interfaces = \{\; &\interface{\upredicateIA{r}}{\upredicateBV{r}}\; \}
\end{aligned}
}%
\]
\centering (c)
\end{minipage}

\end{minipage}

\caption{The \textsf{opp-signs-diff} benchmark.
All variables and constants in $\chcs$ have sort $\sortBV{n}$ for some $n \in \mathbb{N}^+$.
The term $(a \bvxor b) \bvslt 0$ evaluates to true if and only if $a$ and $b$ have opposite signs.}
\vspace{-12pt}
\end{figure}

%% file: benchmarks/swap.tex
\begin{figure}[t]
\centering

% Left column: Panel 1
\begin{minipage}[h]{0.32\textwidth}
\begin{lstlisting}[style=compactC]
void f(int x,int y) {
  assume(x>y);
  assume(y>0);
  int a=0,b=0;
  while(b<y) {
    a++;
    b++;
  }
  while(a<x) {
    a++;
  }
  a=a^b;
  b=b^a;
  a=a^b;
  assert(a<b);
}
\end{lstlisting}
\centering (a)
\end{minipage}
\hfill
% Right column: Panels 2 and 3 stacked vertically
\begin{minipage}[h]{0.67\textwidth}

% Top: Panel 2
\begin{minipage}[h]{\textwidth}
\[
{\footnotesize
\begin{aligned}
\chcs = \{ \;
    & x \bvugt y,\; y \bvugt 0 \; \rightarrow \; p(x,y,0,0), \\
    & p(x,y,a,b),\; b \bvult y \; \rightarrow \; p(x,y,a \bvadd 1,b \bvadd 1), \\
    & p(x,y,a,b),\; \neg(b \bvult y) \; \rightarrow \; q(x,a,b), \\ 
    & q(x,a,b),\; a \bvult x \; \rightarrow \; q(x,a \bvadd 1,b), \\
    & q(x,a,b),\; \neg(a \bvult x) \; \rightarrow \; r(a,b), \\
    & r(a,b),\; a_1 = a \bvxor b,\; b_1 = b \bvxor a_1,\; \neg(a_1 \bvxor b_1 \bvult b_1) \; \rightarrow \; \bot
\; \}
\end{aligned}
}%
\]
\centering (b)
\end{minipage}

\vspace{0.5em} % small vertical gap

% Bottom: Panel 3
\begin{minipage}[h]{\textwidth}
\[
{\footnotesize 
\begin{aligned}
\IAfragment = \{ \;
    & x' > y',\; y' > 0,\; \restrBV{x'},\; \restrBV{y'} \; \rightarrow \; \upredicateIA{p}(x',y',0,0), \\
    & \upredicateIA{p}(x',y',a',b'),\; b' < y' \; \rightarrow \; \upredicateIA{p}(x',y',a'+1,b'+1), \\
    & \upredicateIA{p}(x',y',a',b'),\; \neg(b' < y') \; \rightarrow \; \upredicateIA{q}(x',a',b'), \\ 
    & \upredicateIA{q}(x',a',b'),\; a' < x' \; \rightarrow \; \upredicateIA{q}(x',a'+1,b'), \\
    & \upredicateIA{q}(x',a',b'),\; \neg(a' < x') \; \rightarrow \; \upredicateIA{r}(a',b')
\; \} \\[1ex]
\BVfragment = \{ \;
    & \upredicateBV{r}(a,b),\; a_1 = a \bvxor b,\; b_1 = b \bvxor a_1,\; \neg(a_1 \bvxor b_1 \bvult b_1) \; \rightarrow \; \bot
\; \} \\[1ex]
\interfaces = \{\; &\interface{\upredicateIA{r}}{\upredicateBV{r}}\; \}
\end{aligned}
}%
\]
\centering (c)
\end{minipage}

\end{minipage}

\caption{The \textsf{swap} benchmark.
All variables and constants in $\chcs$ have sort $\sortBV{n}$ for some $n \in \mathbb{N}^+$.
For any $a,b$, the expressions $a_1 \bvxor b_1$ and $b_1$ evaluate to $b$ and $a$, respectively, where $a_1 = a \bvxor b$ and $b_1 = b \bvxor a_1$.}
\end{figure}

%% file: benchmarks/swap_sum.tex
\begin{figure}[t]
\centering

% Left column: Panel 1
\begin{minipage}[h]{0.32\textwidth}
\begin{lstlisting}[style=compactC]
void f(int x,int y) {
  assume(y<=INT_MAX-x);
  assume(y>0);
  int a=0,i=0;
  while(i<x) {
    a++;
    i++;
  }
  i=0;
  while(i<y) {
    a++;
    i++;
  }
  y=y^a;
  a=a^y;
  y=y^a;
  assert(y>=a);
}
\end{lstlisting}
\centering (a)
\end{minipage}
\hfill
% Right column: Panels 2 and 3 stacked vertically
\begin{minipage}[h]{0.67\textwidth}

% Top: Panel 2
\begin{minipage}[h]{\textwidth}
\[
{\footnotesize
\begin{aligned}
\chcs = \{ \;
    & y \bvule (2^n-1) \bvsub x \; \rightarrow \; p(x,y,0,0), \\
    & p(x,y,a,i),\; i \bvult x \; \rightarrow \; p(x,y,a \bvadd 1,i \bvadd 1), \\
    & p(x,y,a,i),\; \neg(i \bvult x) \; \rightarrow \; q(y,a,0), \\
    & q(y,a,i),\; i \bvult y \; \rightarrow \; q(y,a \bvadd 1,i \bvadd 1), \\
    & q(y,a,i),\; \neg(i \bvult y) \; \rightarrow \; r(y,a), \\
    & r(y,a),\; y_1 = y \bvxor a,\; a_1 = a \bvxor y_1,\; \neg(y_1 \bvxor a_1 \bvuge a_1) \; \rightarrow \; \bot
\; \}
\end{aligned}
}%
\]
\centering (b)
\end{minipage}

\vspace{0.5em} % small vertical gap

% Bottom: Panel 3
\begin{minipage}[h]{\textwidth}
\[
{\footnotesize 
\begin{aligned}
\IAfragment = \{ \;
    & y' \leq (2^n-1) - x',\; \restrBV{x'} ,\; \restrBV{y'} \; \rightarrow \; \upredicateIA{p}(x',y',0,0), \\
    & \upredicateIA{p}(x',y',a',i'),\; i' < x' \; \rightarrow \; \upredicateIA{p}(x',y',a'+1,i'+1), \\
    & \upredicateIA{p}(x',y',a',i'),\; \neg(i' < x') \; \rightarrow \; \upredicateIA{q}(y',a',0), \\
    & \upredicateIA{q}(y',a',i'),\; i' < y' \; \rightarrow \; \upredicateIA{q}(y',a'+1,i'+1), \\
    & \upredicateIA{q}(y',a',i'),\; \neg(i' < y') \; \rightarrow \; \upredicateIA{r}(y',a')
\; \} \\[1ex]
\BVfragment = \{ \;
    & \upredicateBV{r}(y,a),\; y_1 = y \bvxor a,\; a_1 = a \bvxor y_1,\; \neg(y_1 \bvxor a_1 \bvuge a_1) \; \rightarrow \; \bot
\; \} \\[1ex]
\interfaces = \{\; &\interface{\upredicateIA{r}}{\upredicateBV{r}}\; \}
\end{aligned}
}%
\]
\centering (c)
\end{minipage}

\end{minipage}

\caption{The \textsf{swap-sum} benchmark.
All variables and constants in $\chcs$ have sort $\sortBV{n}$ for some $n \in \mathbb{N}^+$.
For any $y,a$, the expressions $y_1 \bvxor a_1$ and $a_1$ evaluate to $a$ and $y$, respectively, where $y_1 = y \bvxor a$ and $a_1 = a \bvxor y_1$.
}
\end{figure}

%% file: benchmarks/turn_off_rm1.tex
\begin{figure}[t]
\centering

% Left column: Panel 1
\begin{minipage}[h]{0.32\textwidth}
\begin{lstlisting}[style=compactC]
void f(int x) {
  assume(x>0);
  int y=0,z=x;
  while(y<x-1) {
    y++;
    z--;
  }
  assert(z&(z-1)==0);
}
\end{lstlisting}
\centering (a)
\end{minipage}
\hfill
% Right column: Panels 2 and 3 stacked vertically
\begin{minipage}[h]{0.67\textwidth}

% Top: Panel 2
\begin{minipage}[h]{\textwidth}
\[
{\footnotesize
\begin{aligned}
\chcs = \{ \;
    & x \bvugt 0 \; \rightarrow \; p(x,0,x), \\
    & p(x,y,z),\; y \bvult x \bvsub 1 \; \rightarrow \; p(x,y \bvadd 1,z \bvsub 1), \\
    & p(x,y,z),\; \neg(y \bvult x \bvsub 1) \; \rightarrow \; q(z), \\
    & q(z),\; \neg((z \bvand (z \bvsub 1)) = 0) \; \rightarrow \; \bot
\; \}
\end{aligned}
}%
\]
\centering (b)
\end{minipage}

\vspace{0.5em} % small vertical gap

% Bottom: Panel 3
\begin{minipage}[h]{\textwidth}
\[
{\footnotesize 
\begin{aligned}
\IAfragment = \{ \;
    & x' > 0,\; \restrBV{x'} \; \rightarrow \; \upredicateIA{p}(x',0,x'), \\
    & \upredicateIA{p}(x',y',z'),\; y' < x' - 1 \; \rightarrow \; \upredicateIA{p}(x',y'+1,z'-1), \\
    & \upredicateIA{p}(x',y',z'),\; \neg(y' < x' - 1) \; \rightarrow \; \upredicateIA{q}(z')
\; \} \\[1ex]
\BVfragment = \{ \;
    & \upredicateBV{q}(z),\; \neg((z \bvand (z \bvsub 1)) = 0) \; \rightarrow \; \bot
\; \} \\[1ex]
\interfaces = \{\; &\interface{\upredicateIA{q}}{\upredicateBV{q}}\; \}
\end{aligned}
}%
\]
\centering (c)
\end{minipage}

\end{minipage}

\caption{The \textsf{turn-off-rm} benchmark.
All variables and constants in $\chcs$ have sort $\sortBV{n}$ for some $n \in \mathbb{N}^+$.
Given a bit-vector $z$, the expression $z \bvand (z \bvsub 1)$ clears the rightmost 1-bit of $z$, yielding $0$ if $z$ contains no 1-bits.}
\end{figure}

%% file: benchmarks/turn_on_lsb.tex
\begin{figure}[t]
\centering

% Left column: Panel 1
\begin{minipage}[h]{0.32\textwidth}
\begin{lstlisting}[style=compactC]
void f(int x) {
  assume(x>0);
  int a=0,y=x,b=0;
  while(a<x) {
    a++;
  }
  while(y>0) {
    y--;
    b--;
  }
  assert(((a+b)|1)==1);
}
\end{lstlisting}
\centering (a)
\end{minipage}
\hfill
% Right column: Panels 2 and 3 stacked vertically
\begin{minipage}[h]{0.67\textwidth}

% Top: Panel 2
\begin{minipage}[h]{\textwidth}
\[
{\footnotesize
\begin{aligned}
\chcs = \{ \;
    & x \bvsgt 0 \; \rightarrow \; p(x,0), \\
    & p(x,a),\; a \bvslt x \; \rightarrow \; p(x,a \bvadd 1), \\
    & x \bvsgt 0 \; \rightarrow \; q(x,x,0), \\
    & q(x,y,b),\; y \bvsgt 0 \; \rightarrow \; q(x,y \bvsub 1,b \bvsub 1), \\
    & p(x,a),\; \neg(a \bvslt x),\; q(x,y,b),\; \neg(y \bvsgt 0) \; \rightarrow \; r(a, b), \\
    & r(a,b),\; \neg(((a \bvadd b) \bvor 1) = 1) \; \rightarrow \; \bot
\; \}
\end{aligned}
}%
\]
\centering (b)
\end{minipage}

\vspace{0.5em} % small vertical gap

% Bottom: Panel 3
\begin{minipage}[h]{\textwidth}
\[
{\footnotesize 
\begin{aligned}
\IAfragment = \{ \;
    & x' > 0,\; \restrBV{x'} \; \rightarrow \; \upredicateIA{p}(x',0), \\
    & \upredicateIA{p}(x',a'),\; a' < x' \; \rightarrow \; \upredicateIA{p}(x',a'+1), \\
    & x' > 0 \; \rightarrow \; \upredicateIA{q}(x',x',0), \\
    & \upredicateIA{q}(x',y',b'),\; y' > 0 \; \rightarrow \; \upredicateIA{q}(x',y'-1,b'-1), \\
    & \upredicateIA{p}(x',a'),\; \neg(a' < x'),\; \upredicateIA{q}(x',y',b'),\; \neg(y' > 0) \; \rightarrow \; \upredicateIA{r}(a', b')
\; \} \\[1ex]
\BVfragment = \{ \;
    & \upredicateBV{r}(a,b),\; \neg(((a \bvadd b) \bvor 1) = 1) \; \rightarrow \; \bot
\; \} \\[1ex]
\interfaces = \{\; &\interface{\upredicateIA{r}}{\upredicateBV{r}}\; \}
\end{aligned}
}%
\]
\centering (c)
\end{minipage}

\end{minipage}

\caption{The \textsf{turn-on-lsb} benchmark.
All variables and constants in $\chcs$ have sort $\sortBV{n}$ for some $n \in \mathbb{N}^+$.
Given a bit-vector $x$, the expression $x \bvor 1$ sets the least significant bit of $x$ to~1.}
\end{figure}